\documentclass[trans]{IEEEtran}
\usepackage[T1]{fontenc}
\usepackage{lettrine}
\usepackage{epsfig}
\usepackage{float}
\usepackage{footnote}
\usepackage{cite}
\usepackage{float}
\usepackage{amssymb}
\usepackage{amsthm}
\usepackage{amsmath,mathabx}
\usepackage{color}
\usepackage{epstopdf}
\usepackage{tabularx,tabulary}
\usepackage[font=footnotesize]{caption}
\usepackage{graphicx}
\usepackage{caption}
\usepackage{subcaption}
\usepackage{dblfloatfix}    
\usepackage{bbm}
\usepackage{algorithm,algorithmic}
 \usepackage{multirow}
\definecolor{kirmizi}{RGB}{204,0,0}
\algsetup{linenosize=\small}
\usepackage{cuted}

\newcommand{\E}{\mathbb{E}}
\newcommand*\rfrac[2]{{}^{#1}\!/_{#2}}
\newcommand{\vect}[1]{\boldsymbol{#1}}

\newtheorem{mylemma}{Lemma}

\begin{document}
\title{Distributed Cluster Formation and Power-Bandwidth Allocation for Imperfect NOMA in DL-HetNets}
\author{
Abdulkadir~Celik,~\IEEEmembership{Member,~IEEE,} Ming-Cheng~Tsai,~\IEEEmembership{Student Member,~IEEE,} Redha~M.~Radaydeh,~\IEEEmembership{Senior Member,~IEEE,} Fawaz~S.~Al-Qahtani,~\IEEEmembership{Member,~IEEE,} and~Mohamed-Slim~Alouini,~\IEEEmembership{Fellow,~IEEE}

\thanks{ A. Celik, Ming-Cheng~Tsai, and M-S. Alouini are with Computer, Electrical, and Mathematical Sciences and Engineering Division at King Abdullah University of Science and Technology (KAUST), Thuwal, KSA. R. M. Radaydeh is with the Electrical Engineering Program, Department of Engineering and Technology, Texas A\&M University-Commerce, Commerce, TX, USA. F. S. Al-Qahtani is with Research, Development, Innovation (RDI) Division of Qatar Foundation, Doha, Qatar. A part of this paper has been presented at IEEE GLOBECOM 2018, Singapore \cite{celik2017DLNOMA}. Corresponding author: Abdulkadir Celik (abdulkadir.celik@kaust.edu.sa).
}
\thanks{This work was supported by NPRP from the Qatar National Research Fund under grant no. 8-1545-2-657.}
}
\markboth{To Appear in IEEE Transactions on Communications}{Celik \MakeLowercase{\textit{et al.}}: Distributed Cluster Formation and Power-Bandwidth Allocation for Imperfect NOMA in DL-HetNets}
\maketitle

\begin{abstract}
In this paper, we consider an non-ideal successive interference cancellation (SIC) receiver based imperfect non-orthogonal multiple access (NOMA) schemes whose performance is limited by three factors: 1) Power disparity \& sensitivity constraints (PDSCs), 2) Intra-cluster interference (ICRI), and 3) Intercell-interference (ICI). By quantifying the residual interference with a fractional error factor (FEF), we show that NOMA cannot always perform better than orthogonal multiple access (OMA) especially under certain receiver sensitivity and FEF levels. Assuming the existence of an offline/online ICI management scheme, the proposed solution accounts for the ICI which is shown to deteriorate the NOMA performance particularly when it becomes significant compared to the ICRI. Then, a distributed cluster formation (CF) and power-bandwidth allocation (PBA) approach are proposed for downlink (DL) heterogeneous networks (HetNets) operating on the imperfect NOMA. We develop a hierarchically distributed solution methodology where BSs independently form clusters and distributively determine the power-bandwidth allowance of each cluster. A generic CF scheme is obtained by creating a multi-partite graph (MPG) via partitioning user equipments (UEs) with respect to their channel gains since NOMA performance is primarily determined by the channel gain disparity of cluster members. A sequential weighted bi-partite matching method is proposed for solving the resulted weighted multi-partite matching problem. Thereafter, we present a hierarchically distributed PBA approach which consists of the primary master, secondary masters, and slave problems. For a given cluster power and bandwidth pair, optimal power allocations and Lagrange multipliers of slave problems are derived in closed-form. While power allowance of clusters is updated by the secondary masters based on dual variables of slave problems, bandwidth proportions of clusters are iteratively allocated by the primary master as per the utility achieved by the secondary masters at the previous iteration. Finally, the proposed CF and PBA approaches under the operation of imperfect NOMA are investigated and compared to the OMA scheme by extensive simulations results in DL-HetNets.
\end{abstract}

\maketitle
\newpage
\section{Introduction}
\label{sec:intro}
\lettrine{U}{ltra-Dense} networks have been considered to be a promising solution for the fifth generation (5G) networks as network densification has the ability to boost network coverage and capacity while reducing operational and capital expenditures \cite{wong2017key}. However, traditional HetNets dedicate radio resources to a certain user equipment (UE) either in time or frequency domains, i.e., orthogonal multiple access (OMA), where the number of served UEs at a given time instant is strictly limited by the availability of the radio resources. Considering the expected explosive number of devices, required massive connectivity necessitates more spectrum efficient access schemes with extended coverage.

In this regard, non-orthogonal multiple access (NOMA) has recently attracted attention by permitting to share the same radio resources among a set of UEs \cite{Dai15}, which is also referred to as a NOMA cluster. In particular, NOMA has the capability of providing a higher spectral efficiency while supporting a large number of UEs over the same radio resource. Employing successive interference cancellation (SIC), power domain NOMA can serve multiple UEs at different power levels by ensuring that some UEs can cancel some others' interference out before decoding their own signal. In order to differentiate the desired signal from noise and undecoded signals, the SIC receivers require the disparity of received power levels with a hardware sensitivity gap \cite{Ali2016}, which is referred to as \textit{power disparity and sensitivity constraints} (PDSCs). Moreover, SIC receivers can still observe some residual interference after cancellation due to the propagation of detection and estimation errors, which is often quantified with a \textit{fractional error factor}  \cite{Andrews2005iterative}. In such a case, the ICRI is mainly because of the uncancellable interference and residual interference due to SIC inefficiency. Hence, performance gain achieved by NOMA is primarily limited by imperfections and constraints of SIC receivers and power control policy. Furthermore, cluster formation strategy is an inherently crucial aspect to maximize the benefit offered by NOMA as it is shown that NOMA gain is determined by channel gain discrepancy of cluster members \cite{Ding2015TVT}. Due to its combinatorial nature, CF is a challenging task to accomplish especially in HetNets and necessitates a fast yet high-performance clustering methods. 

In NOMA based downlink (DL) heterogeneous networks (HetNets), all clusters compete for a commonly shared bandwidth whereas clusters within a certain cell contend for the available power of the serving base station (BS). On the other hand, members of a cluster have to share the total power allocated by the BS to their cluster. Taking all these different entities and inter-dependencies into consideration, a centralized CF and PBA scheme require an excessive amount of message passing and coordination among the BSs. To overcome such a communication overhead, it is desirable to have a distributed CF and PBA approach where BSs independently form their own clusters and decide on power and bandwidth allowances, which is the main focus of this paper.

\subsection{Related Works}
\label{sec:rel}

Recent efforts on power domain NOMA can be exemplified as follows:  In \cite{Ali2016}, authors formed clusters based on channel gain ordering and derive closed-form power allocations for a given cluster power and bandwidth pair. The impact of UE selection/clustering is investigated in \cite{Ding2015TVT} for a two-UE DL-NOMA system with fixed and cognitive radio inspired power allocation schemes. The work in \cite{Liu2016fairness} addressed max-min fair UE clustering problem using three different sub-optimal approaches. Authors of \cite{Liu2017Joint} iteratively built clusters where each iteration jointly optimize beam-forming and power allocation for given clusters. Another work considered beam-forming and power allocation of a multiuser multiple-input-multiple-output (MIMO) NOMA system where two-UE clusters are formed from high and low channel gain UEs with the consideration of channel gain correlations \cite{Ali2017beam}. We investigate cluster formation and resource allocation problems for DL and UL HetNets in \cite{celik2017DLNOMA} and \cite{celik2017ULNOMA}, respectively. 

Joint power and channel allocation for the NOMA system are addressed in \cite{Lei2016} wherein a near optimal solution was proposed by combining Lagrangian duality and dynamic programming. By using Lyapunov optimization framework, the short and long-term network utility is maximized by joint data rate and power control in \cite{Bao2017joint}. In \cite{Elbamby2017}, authors study the problem of resource optimization, mode selection and power allocation in wireless cellular networks under the assumption of full-duplex NOMA capability and queue stability constraints. Sun et. al. considered joint power and subcarrier allocation for full-duplex multi-carrier (MC) NOMA systems for UL and DL transmission of a single BS \cite{Sun2017optimal}. MC-NOMA is also studied in \cite{Wei2017optimal} where authors jointly design the power and rate allocation, user scheduling, and successive SIC decoding policy to minimize the power consumption. In \cite{Zhu2017optimal}, the power is controlled to achieve the different objective for given channel allocations. Sub-Channel assignment, power allocation, and user scheduling are addressed by formulating the sub-channel assignment problem as equivalent to a many-to-many two-sided user-subchannel matching game \cite{Di2016}.  By only utilizing a single scalar, an $\alpha$-fairness approach is developed to achieve different UE fairness levels in \cite{Xu2017optimal}.  In\cite{Abbasi2017resource}, authors investigated resource allocation for hybrid NOMA system subject to proportional rate constraints. The work in \cite{Xu2017resource} focused on resource allocation in energy-cooperation enabled two-tier NOMA HetNets with energy harvesting BSs. Authors of \cite{Zhao2017spectrum} allocated spectrum and power using a many-to-one matching game and sequential convex programming, respectively. 

  Since it can cause severe performance degradation, a practical design of NOMA must account for real-life imperfections which can be a result of imperfect channel state information (CSI), residual interference due to the FEF, or PDSCs of SIC receivers. In \cite{Yang2016}, authors investigate the impact of partial CSI on the performance of the NOMA networks. They first consider an imperfect CSI model where the BS and UEs have an estimate of the channel and a priori knowledge of the variance of the estimation error. Analytical and numerical findings demonstrated that the average sum rate of NOMA systems can always outperform conventional OMA. Based on the second order statistics, authors show that NOMA is still superior to conventional OMA. In opportunistic one-bit feedback has been used for NOMA in \cite{Xu2016} where a closed-form expression for the common outage probability is derived along with the optimal diversity gains under short and long-term power constraints. As discussed in \cite{SIC_OFDMA}, the SIC receiver performance is mainly determined by the tradeoff between computational complexity and error propagation during the interference cancellation (IC) process. Albeit their significant contributions, to the best of authors' knowledge, none of the aforementioned works consider the residual interference due to the SIC error propagation. Excluding \cite{Ali2016}, these works also do not take the PDSCs and cluster formation design into account. Furthermore, challenges of the HetNet environment is only addressed in \cite{Xu2017resource, Zhao2017spectrum} where authors do not consider a distributed approach. To the best of our knowledge, this paper is the first work to consider an imperfect NOMA with residual interference and to develop a distributed CF and PBA for DL-HetNets.
\subsection{Main Contributions}
\label{sec:cont}
Our main contributions can be summarized as follows:
\begin{itemize}
\item 
In practice, constraints and imperfections of SIC receivers constitute a limiting factor on the achievable gain by NOMA. This work is the first to consider the impacts of residual interference on NOMA performance due to the non-ideality of the SIC receivers. Although our work is not aimed at proposing an ICI management scheme, the proposed power allocation method is also capable of accounting for the leftover ICI from any offline/online ICI management scheme. Obtained results demonstrate that NOMA cannot always perform better than the OMA under certain FEF levels and receiver sensitivity values. It is also shown that being agnostic to the ICI can severely degrade the performance especially when it becomes significant in comparison with the ICRI, which clearly indicates the necessity for an effective ICI management scheme.  

\item  
After formulating a centralized CF and PBA as an mixed-integer non-linear programming (MINLP) problem, we develop a distributed solution methodology where BSs independently form clusters and determine the power-bandwidth allowance of each cluster. Noting that existing solutions have contended to basic NOMA clusters of size two, a generic CF scheme is obtained by creating a \textit{multi-partite graph} (MPG) via partitioning UEs with respect to their channel gains.  A sequential \textit{weighted bi-partite matching} (WBM) method is proposed for solving the resulted \textit{weighted  multi-partite matching} (WMM) based CF problem in cubic order. If edges are merely weighted by the channel gain disparity of UEs, it is proven that the complexity of solving the WMM can even be reduced to quasi-linear order without executing any matching algorithm. Obtained results show that proposed CF delivers a performance very close to the exhaustive centralized solution with a significantly reduced processing load. 


\item 
By employing primal and dual decomposition methods, we propose a hierarchically distributed PBA approach which consists of a primary master problem, secondary master problems, and slave problems. For a given cluster power and bandwidth allowance, optimal power allocations and Lagrange multipliers of imperfect NOMA are derived in closed-form subject to PDSCs and QoS constraints. Closed-form solutions are used by secondary master problems to update total power allowance of clusters. Based on achieved cluster utilities, the primary master problem iteratively updates the bandwidth allowances to maximize the network utility and broadcasts updated bandwidth allocations to the BSs. Finally, we show that proposed algorithm greatly reduces the communication overhead and investigate the NOMA performance in comparison with OMA under different BS density, traffic offloading bias factor, UE density, and cluster size scenarios for DL-HetNets. 
\end{itemize}

\begin{table}
\begin{center}
\scriptsize
\begin{tabular}{ |l|l| }
  \hline
  \multicolumn{2}{|c| }{\textbf{Table of Notations}} \\
  \hline
  \textbf{Not.} & \textbf{Description} \\ 
\hline
$\mathcal{C}$ & Set of $C$ BSs, $\mathcal{C} \triangleq \{c | \: 0 \leq c \leq S \}$ where $c=0$ is the MBS.  \\ 
\hline
$\mathcal{U}$ & Set of $U$ UEs, $\mathcal{U} \triangleq \bigcup_c \mathcal{U}_c$ where $\mathcal{U}_c$ is the set of $U_c$ UEs of BS$_c$  \\
\hline
$\mathcal{R}$ & Set of all clusters, $\mathcal{R} \triangleq \bigcup_c \mathcal{R}_c$ where $\mathcal{R}_c$ is the set of $R_c$ clusters.\\ 
\hline
$\mathcal{K}_c^r$ & Set of users belong to $r^{th}$ cluster of BS$_c$, $1 \leq r \leq R_c$, $1 \leq c \leq C$.\\ 
\hline
$P_c$ & BS transmission power, i.e., $P_c=P_m/P_c=P_s$ for MBS/SBSs.\\ 
\hline
$N_0$ & Noise power spectral density.\\ 
\hline
$B$ & Total available bandwidth for all UEs.\\  
\hline
$\alpha_{c,r}^i$ & Binary variable for cluster membership relations.\\ 
\hline
$\varpi_c^r$ & Total power fraction allocated for clusters, $\sum_r \varpi_c^r \leq 1.$\\ 
\hline
$\omega_{c,r}^i$ & Power allocation variable, $\omega_{c,r}^i \leq \alpha_{c,r}^i$, $\sum_{i \in \mathcal{K}_c^r} \omega_{c,r}^i \leq \varpi_c^r$.\\ 
\hline
$\theta_c^r$ & Bandwidth allocated to UEs within $\mathcal{K}_c^r$, $\theta_c^r \in [0,1]$, $\sum_{c,r}\theta_c^r \leq 1$. \\ 
\hline
$\gamma_{c,r}^i$ & SINR of UE$_i \in \mathcal{K}_c^r$.\\ 
\hline
$\epsilon_i$ & SIC error factor of UE$_i$.\\ 
\hline
$p_\Delta$ & Receiver sensitivity og UE$_i$.\\ 
\hline
$g_c^i$ & Composite channel gain from BS$_c$ to UE$_i$.\\ 
\hline
$\mathbb{C}_{c,r}^i$ & Achievable capacity of UE$_i$ which requires $\mathbb{C}_{c,r}^i \geq \bar{\mathbb{C}}_{c,r}^i$\\ 
\hline
$L_c$ & Affordable number of cancellations.\\ 
\hline
$R_c$ & Number of clusters of BS$_c$, $R_c=\lceil U_c/ L_c \rceil$.\\ 
\hline
$\mathcal{P}_c^\ell$ & Set of UE partitions of BS$_c$, $0 \leq \ell \leq L_c$.\\ 
\hline
$\mathcal{E}_{c,l}^{r,s}$ & Edge weights for multi-partite graphs.\\ 
\hline
$\mathcal{L}_c^r$ & Lagrange function related to $\mathcal{K}_c^r$.\\ 
\hline
$\lambda_c^r$ & Lag. multiplier related to cluster power consumption.\\ 
\hline
$\mu_{c,r}^i$ & Lag. multiplier related to QoS constraint of UE$_i$ \\ 
\hline
$\varphi_{c,r}^i$ & Lag. multiplier related to PDSC of UE$_i$. \\ 
\hline
$\nu/\upsilon$ & Step size for updates of distributed algorithm.\\ 
\hline
$\rho_{c,r}^i$ & Composite parameter defined as $\rho_{c,r}^i \triangleq N_0 B \theta_c^r (P_c g_c^i)^{-1}$.\\ 
\hline
$q_{c,r}^i$ & Composite parameter defined as $q_{c,r}^i \triangleq 2^{\bar{\mathbb{C}}_{c,r}^i / B \theta_c^r}$.\\ 
\hline
$\Delta_{c}^i$ & Composite parameter defined as $\Delta_c^i \triangleq p_\Delta (P_c g_c^i)^{-1}$.\\ 
\hline
\end{tabular}
\caption{Table of Notations}
\label{Ts:notations}
\end{center}
\end{table}
\normalsize

\subsection{Notations and Paper Organization}
\label{sec:org}
Throughout the paper, sets and their cardinality are denoted with calligraphic and regular uppercase letters (e.g., $| \mathcal{A}|=A$), respectively. Vectors and matrices are represented in lowercase and uppercase boldfaces (e.g., $\vect{a}$ and $\vect{A}$), respectively. Superscripts $c$, $r$, and $i$ are used for indexing BSs/cells, clusters, and UEs, respectively. The optimal values of variables are always marked with superscript $\star$, e.g., $x_{c,r}^{i,\star}$ and $y_{c,r}^\star$. 

The remainder of the paper is organized as follows: Section \ref{sec:sys_mod} introduces the system model along with constraints and imperfections of SIC. Section \ref{sec:CF-PBA} first formulates the optimal CF and PBA problem. Section \ref{sec:CF} and Section \ref{sec:PBA} address proposed CF and distributed PBA methods, respectively. Numerical results are presented in Section \ref{sec:res} and Section \ref{sec:conc} concludes the paper with a few remarks.

\begin{figure}[t!]
    \centering
        \includegraphics[width=0.35\textwidth]{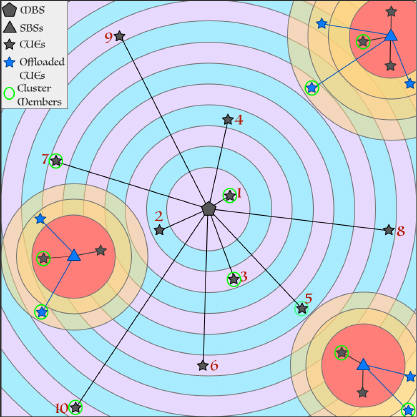}
        \caption{Illustration of clustering and traffic offloading in DL-HetNets.}
        \label{fig:sys}
\end{figure}

\section{System Model}
\label{sec:sys_mod}
\subsection{Network Model}
We consider DL transmission of a 2-tiered HetNet where each tier represents a particular cell class, i.e., tier-1 consists of a single macrocell and tier-2 comprises of smallcells as shown in Fig. \ref{fig:sys}. Denoting the number of small BSs (SBSs) as $S$, the index set of all BSs is represented by $\mathcal{C}=\{c | \: 0 \leq c \leq S \}$ where $c=0$ and $1 \leq c \leq S$ represent the macro BS (MBS) and SBSs, respectively. We note that the terms BS, cell and their indices are used interchangeably throughout the paper. Maximum transmission powers of BSs are generically denoted as $P_c$ which equals to $P_m$ and $P_s$ for the MBS and SBSs, respectively. Furthermore, index set of all $U \triangleq \sum_c U_c$ UEs is given as $\mathcal{U}\triangleq\bigcup_c \mathcal{U}_c$ where $\mathcal{U}_c$ is the set of $U_c$ UEs associated with BS$_c$. UE-BS association is based on received signal strength (RSS) information with a certain offloading bias factor \cite{Celik2017Resource,celik2017joint}. A simple traffic offloading example is shown in Fig. \ref{fig:sys} where yellow-colored circles represent the offloading regions corresponding to different  bias factors and blue-colored stars represent the UEs offloaded from the MBS. Likewise, the set of all $R$ clusters are denoted as $\mathcal{R}\triangleq \bigcup_c \mathcal{R}_c$ where $\mathcal{R}_c$ is the set of $R_c$ clusters of BS$_c$. Hence, $\mathcal{U}_c$ is partitioned into $R_c$ clusters such that $\mathcal{K}_c^r$ symbolizes the set of $K_c$ UEs within cluster $r$, that is, $\mathcal{U}_c=\bigcup_{r\in \mathcal{R}_c} \mathcal{K}_c^r$ and  $R_c=\left \lceil \frac{U_c}{K_c} \right \rceil$. Cluster $\mathcal{K}_c^r$ is allowed to utilize $\theta_c^r \in [0,1]$ portion of the entire DL bandwidth, $B$, $\forall c, r$. Since each cluster has its own dedicated bandwidth, total number of clusters and bandwidth proportions are equivalent. Also noting that all $R_c$ clusters of cell $c$ share the DL transmission power of the BS$_c$, power fraction allocated for $\mathcal{K}_c^r$ is defined as $\varpi_c^r \in [0,1]$, $\sum_{r\in \mathcal{R}_c}\varpi_c^r \leq 1, \: \forall c$. 

\subsection{Imperfections and Constraints of SIC}
\label{sec:SIC}
The SIC receiver first decodes the stronger interferences and then subtracts them from the broadcasted signal until they obtain the desired signal. Accordingly, the received interference strengths are required to be sufficiently higher in comparison to the intended signal for a  successful IC process. Accordingly, BSs broadcast the superposed signals with low power level for high channel gain UEs and high power level signals for low channel gain UEs. In this case, the highest channel gain UE can cancel all the interference while being allocated to the lowest power level. On the other hand, the lowest channel gain UE cannot cancel any interference while being allocated to the highest power level. In such a way that the performance of the entire cluster is enhanced in a fair manner. 

To be more specific, let us now focus on cluster $r$ of BS$_c$, $\mathcal{K}_c^r = \{i | \:\alpha_{c,r}^i=1, \: g_c^{i-1}\geq g_c^i \geq g_c^{i+1}, \forall i \}$ where $\alpha_{c,r}^{i} \in \{0,1\}$ is a binary indicator for the cluster membership and cluster members are sorted in the descending order of the channel gains, $g_c^i$, without loss of generality. Such clusters are demonstrated in Fig. \ref{fig:sys} with green circles around the UEs, for example, first cluster of the MBS is $\mathcal{K}_0^1=\{1,3,5,7,10\}$. As per SIC principles, NOMA allocates transmission power weights as $\omega_{c,r}^{i-1} < \omega_{c,r}^i < \omega_{c,r}^{i+1}, \: \forall i \in \mathcal{K}_c^r$, hence, normalized received power at UE$_i$ is given as
\begin{equation}
\label{eq:order}
\overbrace{\underbrace{  \omega_{c,r}^{K_c^r} g_c^i > \ldots > \omega_{c,r}^{i+1} g_c^i > }_{\text{Higher Rank Decoding Order: }\mathcal{O}_i^h} }^{\text{Cancellable Signals}} \overbrace{\underbrace{\omega_{c,r}^{i} g_c^i }_{\text{Signal}}}^{\text{Desired}} \overbrace{\underbrace{> \omega_{c,r}^{i-1} g_c^i > \ldots > \omega_{c,r}^{i} g_c^i }_{\text{Lower Rank Decoding Order: }\mathcal{O}_i^{\ell}}}^{\text{Uncancellable Signals}}
\end{equation}
where normalization is with respect to BS transmission power $P_c$, $\mathcal{O}_i^h=\{i+1, \ldots, K_c^r \}$ is the higher rank decoding order set, and $\mathcal{O}_i^{\ell}=\{1, \ldots, i\}$ is the lower rank decoding order set for UE$_i$. UE$_i$ can only cancel the interference induced by higher rank members, while interference from lower rank members cannot be decoded as they are weaker than the desired signal. Furthermore, the hardware sensitivity of the SIC receivers requires a minimum signal-to-interference-plus-noise-ratio (SINR) to distinguish intended signals from noise. Therefore, \textit{power disparity \& sensitivity constraints }(PDSCs) can be expressed in linear scale as \cite{Ali2016} 
\begin{equation}
\label{eq:PDC}
P_c\omega_{c,r}^{j} g_c^i - \sum_{k=i}^{j-1} P_c \omega_{c,r}^{k} g_c^i  \geq p_{\Delta},\textbf{  }\forall \text{ UE}_i,  \forall \text{ UE}_j  \in \mathcal{O}_i^h
\end{equation}
where $p_{\Delta}$ denotes the hardware sensitivity. The intuition behind (\ref{eq:PDC}) is that during the IC process of UE$_j \in\mathcal{O}_i^h$, receiver of $\textrm{UE}_i$ observes undecoded signals of UE$_k \in\mathcal{O}_i^h$, $j>k$, as noise. Moreover, a non-ideal SIC observes some residual interference due to the error propagation which is caused by detection and estimation errors \cite{SIC_OFDMA}. Accordingly, a generic SINR representation of the imperfect SIC receiver is given  as 
\begin{align}
\label{eq:sinr_c_dl}
\gamma_{c,r}^{i}& =\frac{\omega_{c,r}^{i} \alpha_{c,r}^i}{  \sum_{ \substack{ l \in \mathcal{O}_{i}^{\ell} \\ l \in \mathcal{K}_c^r}} \omega_{c,r}^{l} \alpha_{c,r}^l+ \epsilon_i \sum_{ \substack{ j \in \mathcal{O}_{i}^{h} \\j \in \mathcal{K}_c^r}} \omega_{c,r}^{j} \alpha_{c,r}^j
+\rho_{c,r}^i}
\end{align} 
where the first two terms of the denominator characterize the ICRI, $0 \leq \epsilon_i \leq 1$ is the fractional error factor (i.e., $1-\epsilon_i$ can be regarded as the SIC efficiency) which determines the residual interference after the IC, $\rho_{c,r}^i \triangleq \frac{I_{ici}+N_0 B \theta_c^r} {P_c g_c^i}$, $I_{ici}$ is ICI generated by UEs located at different macrocell coverage area and cannot be canceled by the SIC receiver \footnote{Notice that a dominant ICI can decrease the SIC efficiency (i.e., increase $\epsilon$) due to the failure in the decoding process.}, $B$ is the entire DL bandwidth, and $N_0$ is the thermal noise power spectral density. We must note that the  ICI can have a significantly negative impact on the NOMA performance in Het-Nets. Indeed, it can be handled by adopting the existing interference management schemes of traditional OMA networks, which can be classified as offline (e.g., fractional frequency reuse, sectorization or spatial reuse) and online (silencing, dynamic spectrum access, cooperative beamforming, cooperative multi-point transmission,etc.) techniques \cite{Ali2017, Liu2017}. Although the remainder of the paper factors the ICI in the optimal power allocation scheme, our contributions do not include developing ICI management schemes. Hence, the capacity achievable by UE$_i$ is given by 
\begin{equation}
\label{eq:capacity}
\mathbb{C}_{c,r}^i = B \sum_{c,r}  \theta_c^r  \log_2(1+ \gamma_{c,r}^i).
\end{equation}
Even if this work assumes a perfect CSI estimation, uncertainty about CSI is crucial from two aspects: First, power allocations obtained based on imperfect CIS can considerably degrade the NOMA performance due to the increasing impact of intra-cluster interference. Second, although small and fast fading component may not have a significant effect on the channel gain order, large-scale fading should still be estimated with sufficient accuracy to reduce the impacts of CSI estimation errors on the clustering strategy\footnote{Even though it is out of this paper's scope, extending the proposed methods to a robust optimization framework is necessary to account for CSI estimation errors.}.

\section{Cluster Formation and Power Bandwidth Allocation}
\label{sec:CF-PBA}
In a HetNet, determining optimal values for integer-valued cluster numbers/sizes and binary valued UE-cluster associations yields high computational complexity. Furthermore, highly non-convex nature of PBA problem induces an extra complexity to achieve optimal performance. A generic CF policy may consider all UEs as potential candidates for all clusters, which has a high time complexity in the order of $\mathrm{O}\left(2^{U \times R}\right)$ for an exhaustive CF solution. Since clusters can be formed among UEs associated with different BSs, such an approach also necessitates a low-latency and robust coordination among the BSs in order to decide on clusters and optimal transmission power levels, which requires the exchange of channel gains and decoded signals to perform SIC locally and naturally yields a high communication overhead.  However, we consider a more practical scenario where BSs independently form clusters among their UEs such that the order of time complexity can be reduced to $\mathrm{O}\left(2^{U_c \times R_c}\right), \forall c $.  In this way, a distributed CF and PBA approach can be developed as each BS is responsible for CF and PBA of its own UEs. In what follows, we first formulate a centralized joint CF and PBA problem and then present the proposed distributed solution.

\subsection{Centralized CF and PBA}
 The centralized CF and PBA problem can be formulated as in $ \vect{\mathrm{P}_\mathrm{o}}$ where we make the following assumptions: 1) A UE can be associated with exactly one cluster at a time and  2) The cluster size is determined by ensuring that the highest channel gain UE can cancel all other cluster members. In particular, the first assumption enables us to decompose the entire problem into sub-problems and develop a hierarchically distributed PBA scheme. On the other hand, BS$_c$ locally allocates a certain power fraction of its available DL transmission power to its own clusters, $\varpi_c^r, \forall c, r$.

\begin{figure}[htbp!]
\begin{equation*}
\hspace*{0pt}
 \begin{aligned}
 & \hspace*{0pt} \vect{\mathrm{P}_\mathrm{o}}: \underset{\vect{\alpha},\vect{\theta},\vect{\omega}, \vect{\varpi}}{\max}
& & \hspace*{3 pt} B \sum_{c} \sum_{r \in \mathcal{R}_c}  \theta_c^r  \sum_{i \in \mathcal{K}_c^r}\log_2(1+ \gamma_{c,r}^i) \\
& \hspace*{0pt} \mbox{$\mathrm{C}_o^1$: }\hspace*{8pt} \text{s.t.}
&&  \sum_{r} \alpha_{c,r}^i = 1, \textbf{ } \hspace{56 pt} \forall c,i\\ 
 &
 \hspace*{0 pt}\mbox{$\mathrm{C}_o^2$: } & & \sum_{ i } \alpha_{c,r}^i\leq K_c, \textbf{ } \hspace{50 pt} \forall c,r \\
    &
\hspace*{0 pt}\mbox{$\mathrm{C}_o^3$: } & & \sum_{i} \omega_{c,r}^i \leq \varpi_c^r,\textbf{ } \hspace{48pt} \forall c, \forall r \\
    &
    \hspace*{0 pt}\mbox{$\mathrm{C}_o^4$: } & & \sum_{r} \varpi_c^r \leq 1,\textbf{ } \hspace{60pt} \forall c\\
    &
    \hspace*{0 pt}\mbox{$\mathrm{C}_o^5$: } & & \sum_{r} \theta_c^r \leq 1,\\
&
  \hspace*{0 pt}\mbox{$\mathrm{C}_o^6$: } & & 0 \leq \omega_{c,r}^i \alpha_{c,r}^i - q_{c,r}^i \left(\sum_{j=1}^{i-1}\omega_{c,r}^j\alpha_{c,r}^j \right.\\
  &&& \left. +\epsilon_i \sum_{k=i+1}^{K_c^r} \omega_{c,r}^k \alpha_{c,r}^j+ \rfrac{\beta_c^r}{\rho_{c,r}^i}\right),  \textbf{ } \hspace{0 pt}\forall c,\forall r,\forall  i  \\
&
  \hspace*{0 pt}\mbox{$\mathrm{C}_o^7$: } & & 0 \leq \omega_{c,r}^i \alpha_{c,r}^i - \sum_{j=1}^{i-1}\omega_{c,r}^j \alpha_{c,r}^j-\Delta_c^i ,  \textbf{ } \hspace{2 pt}\forall i>1 \\
  &
\hspace*{0 pt}\mbox{$\mathrm{C}_o^{8}$: } & &  \omega_{c,r}^i, \varpi_c^r, \theta_c \in [0,1], \: \alpha_{c,r}^i \in \{ 0,1\},
\end{aligned}
\end{equation*}
\end{figure}
In $\vect{\mathrm{P}_\mathrm{o}}$, $\mathrm{C}_o^1$ ensures that a UE is assigned to only one cluster and $\mathrm{C}_o^2$ puts an upper bound on the number of UEs within a cluster by $K_c$ which is determined based on the affordable number of ICs by UEs that is a design parameter and denoted by $L_c$. While $\mathrm{C}_o^3$ limits the power consumption of $\mathcal{K}_c^r$ to the fraction $\varpi_c^r$, $\mathrm{C}_o^4$ restricts the total amount of clusters' power weights to available BS$_c$ power. Likewise, $\mathrm{C}_o^5$ constrains the sum of cluster bandwidths to the available total bandwidth, $B$. $\mathrm{C}_o^6$ introduces the QoS requirements where $q_{c,r}^i\triangleq2^{\rfrac{\overline{\mathbb{C}}_{c,r}^i}{B\theta_c^r}}$ and $\overline{\mathbb{C}}_{c,r}^i$ is the data rate demand of UE$_i$\footnote{ $q_{c,r}^i$ is obtained from \ref{eq:capacity} which reduces a single term due to the $\mathrm{C}_o^1$.}. Since clusters are assumed to have dedicated spectrum, cluster index $r$ also represents the subcarrier indexes. $\mathrm{C}_o^6$ present PDSCs where $\Delta_c^i \triangleq \rfrac{p_{\Delta}}{P_c g_c^i}$. Finally, variable domains are defined in $\mathrm{C}_o^{8}$ where the power allocation for UE$_u$ on cluster $r$ is set to zero ($0 \leq\omega_{c,r}^u \leq  \alpha_{c,r}^u$) if UE$_u \notin \mathcal{K}_c^r$.
 
 $\vect{\mathrm{P}_\mathrm{o}}$ is apparently a MINLP problem which requires impractical time complexity even for moderate sizes of HetNets. As a fast yet high performance suboptimal solution methodology is of the essence to employ NOMA in practice, we  develop a solution methodology by decoupling this hard problem into CF and PBA subproblems, which is desirable to obtain a tractable solution methodology and also durable as it is already shown that the NOMA performance gain is primarily determined by channel gain disparity of the cluster members \cite{Ding2015TVT}.

\subsection{Distributed CF and PBA}
 
Let us first classify the type of network resources based on the following hierarchical relationship: The bandwidth proportion $\theta_c^r$ is a \textit{primary-global} resource in which BSs compete for their clusters. Therefore, $\theta_c^r$ and $\sum_{c,r} \theta_c^r \leq 1$ are primary-global coupling variable and complicating constraint among the BSs, respectively. Likewise, as each BS has its own power source, the fraction of total transmission power $\varpi_c^r$ is a \textit{secondary-global} resource in which clusters contend for their members. Hence, $\varpi_c^r$ and $\sum_r \varpi_c^r \leq 1$ are secondary-global coupling variable and complicating constraint for BS$_c$ clusters, respectively. On the other hand, $\omega_{c,r}^i$ is a \textit{local} resource needed by cluster members UE$_i \in \mathcal{K}_c^r$ to satisfy PDSCs and QoS constraints. Thus, $\omega_{c,r}^i$ and $\sum_{i \in \mathcal{K}_c^r} \omega_{c,r}^i \leq \varpi_c^r$ are secondary-global coupling variable and complicating constraint for cluster members, respectively.

 \begin{figure}[t!]
    \centering
        \includegraphics[width=0.35 \textwidth]{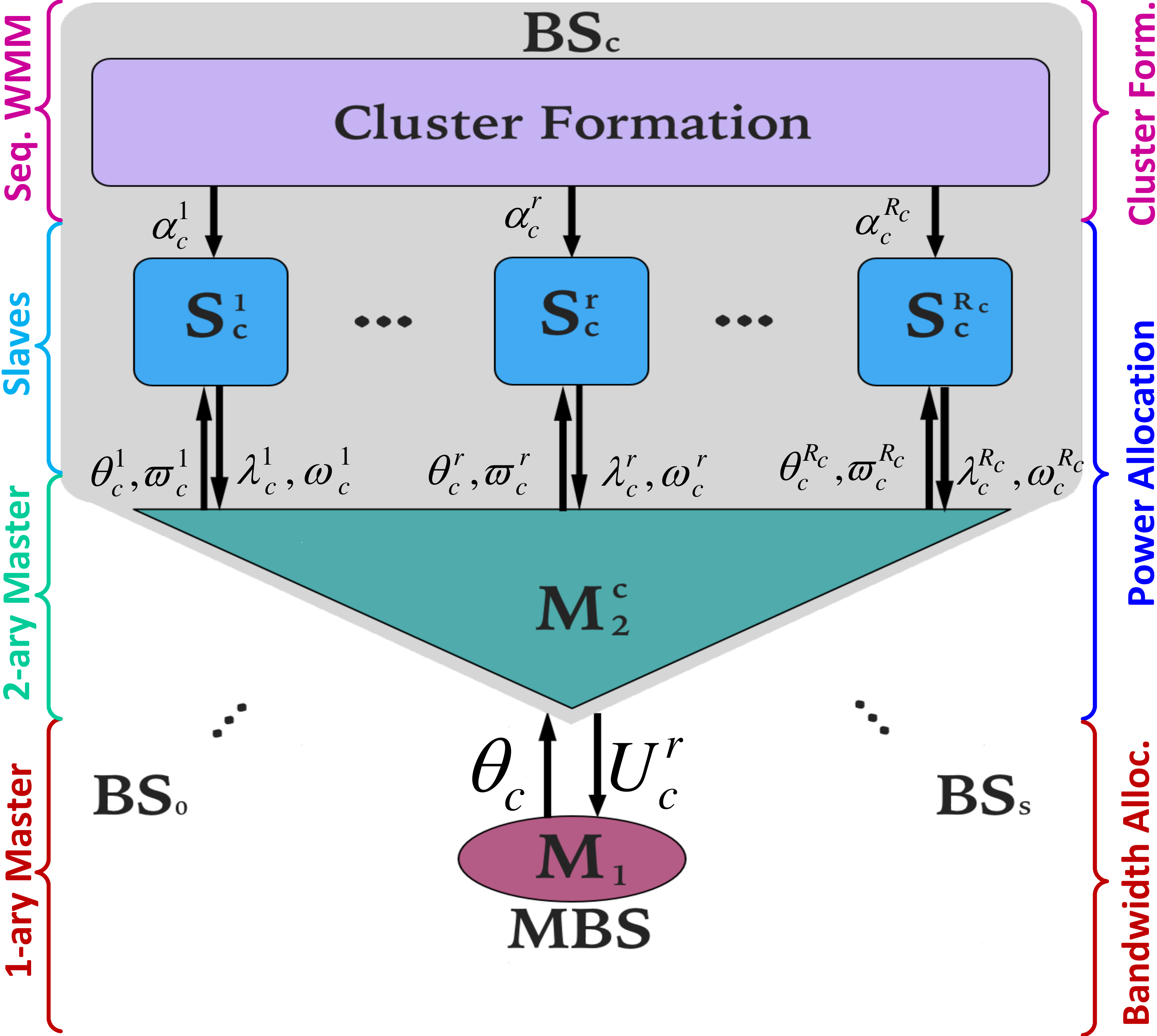}
        \caption{Schematic illustration of the proposed hierarchically distributed solution.}
        \label{fig:distributed}
\end{figure} 

Exploiting this hierarchical relation and decomposability of the problem, CF and PBA problem can be solved in a distributed manner as shown in Fig. \ref{fig:distributed} where each BS can form its own clusters by decoupling binary clustering variable $\vect{\alpha}$ from power-bandwidth variables because the achievable performance gain of NOMA has shown to be mainly determined by the channel gain disparity of the cluster members \cite{Ding2015TVT}. Thereafter, decomposition methods is applied to given cluster formations in order to conclude the optimal power and bandwidth portions in a distributed manner as follows: A central unit (preferably the MBS) is responsible for \textit{primary-master} problem $\vect{\mathrm{M}_1}$ which decides on bandwidth portions $\vect{\theta}$ as per the utility achieved by each cluster. Depending upon the bandwidth allocated by $\vect{\mathrm{M}_1}$, $\vect{\theta_c}$, BS$_c$ is responsible for its own \textit{secondary-master} problem $\vect{\mathrm{M}_2^c}$ in order to determine the power fraction $\vect{\varpi_c}$ for clusters of BS$_c$, $\forall c$. Finally, $\vect{\mathrm{M}_2^c}$ is further divided into $R_c$ \textit{slave} problems which maximize the total cluster utility for a certain bandwidth $\theta_c^r$ and power fraction $\varpi_c^r$ given by primary and secondary masters, respectively. 

\section{Weighted Multi-Partite Matching Based CF}
\label{sec:CF}
In our proposed solution, BS$_c$ independently forms $R_c$ clusters among UEs $ \in \mathcal{U}_c$ in a distributed fashion. Therefore, let us omit the cluster indices and consider the CF of a single BS without loss of generality. Denoting the affordable number of ICs as $L_c$, the maximum size of clusters is limited by $K_c=L_c+1$ as the highest channel gain UE within a cluster can cancel interference of at most $L_c$ UEs. From multi-user detection theory it is known that the capacity region of NOMA improves with the channel gain disparity of users. Since the performance gap between NOMA and OMA schemes are primarily determined by the channel gain disparity of cluster members \cite{Ding2015TVT}, the proposed clustering method tries to maximize overall channel gain disparity.
 
Accordingly, we develop a \textit{matching theory} based clustering algorithm by partitioning the UE index set of BS$_c$, $\mathcal{U}_c=\{i \: \vert \: g_c^i > g_c^{i+1}, 1\leq i \leq U_c \}$, into $L_c+1$ disjoint channel gain levels, i.e., 
\begin{align}
\label{eq:partit}
\mathcal{P}_c^\ell = \{ i \: \vert \: i \in \mathcal{U}_c; \: (\ell-1) R_c +1 \leq i \leq \min(\ell R_c, U_c) \},
\end{align}
where $1 \leq \ell \leq K_c $, $R_c= \lceil U_c/L_c\rceil$ is the number of clusters/sub-carriers of $\mathcal{U}_c$ as a function of the cluster size and $\mathcal{P}_c^{\ell'}  \cap \mathcal{P}_c^{\ell} = \emptyset, \: \forall \ell'\neq \ell$. That is, intra-partition and inter-partition channel gains are in descending order (i.e., $g_c^i \geq g_c^j, \: i,j \in \mathcal{P}_c^\ell, \: \forall i < j $), and the lowest channel gain within $\mathcal{P}_c^\ell$ is higher than all channel gains within  $\mathcal{P}_c^{\ell'}, \: \forall \ell' > \ell$. Notice that the partitioning in \eqref{eq:partit} requires sort of channel gains and has therefore a time complexity of $\mathrm{O}(U_c \log U_c)$. Pictorially, one can also think of partitions as UEs falling into non-overlapping spatial ring zones in a free-space path loss channel model as shown in Fig. \ref{fig:sys}. 

\begin{figure}[!t]
    \centering
        \includegraphics[width=0.45\textwidth]{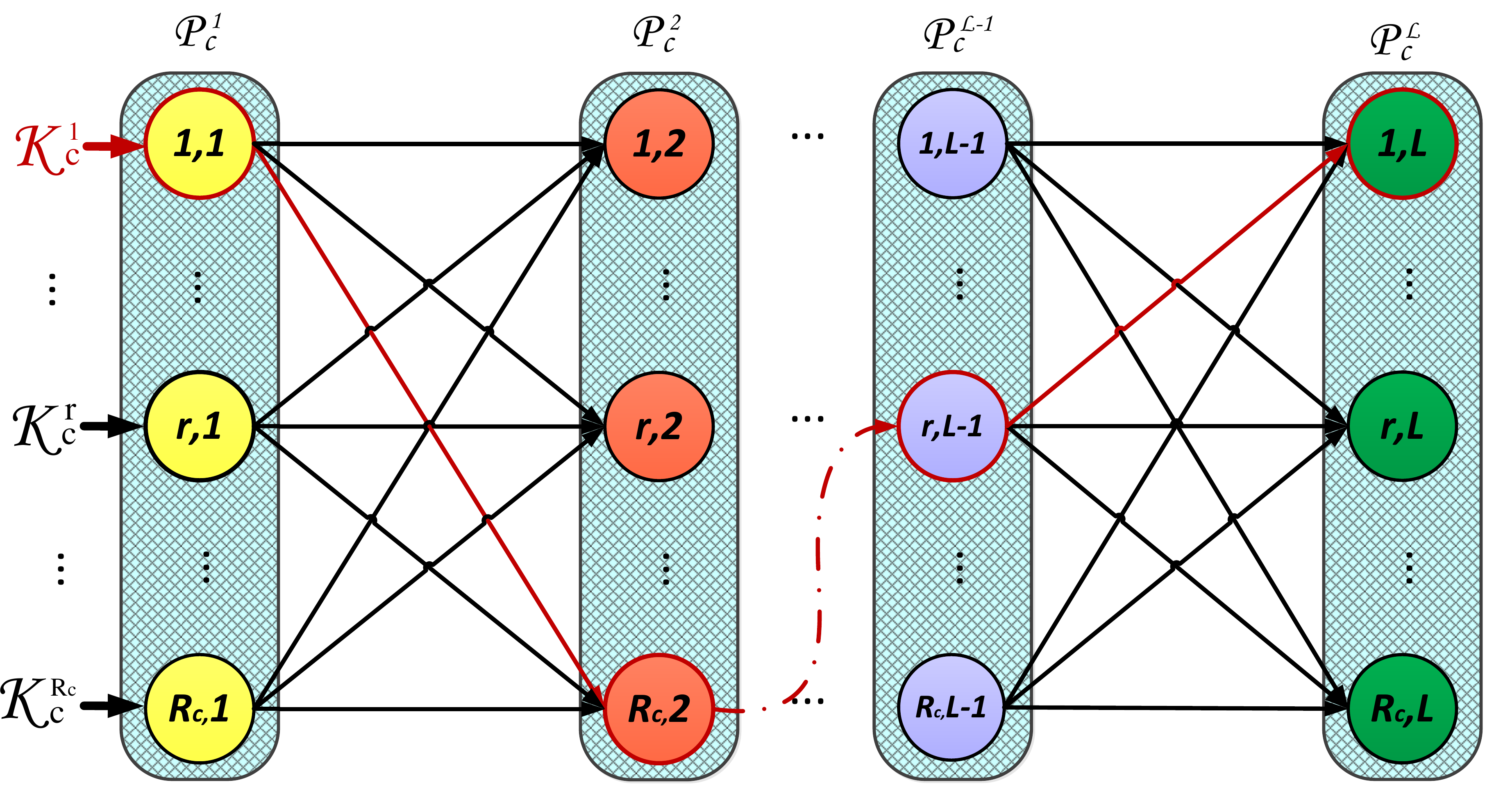}
        \caption{Weighted multi-partite matching representation for CF.}
        \label{fig:mpartite}
\end{figure}

The partitioning yields a multi-partite graph (MPG) as depicted in Fig. \ref{fig:mpartite} where edge weight from $i^{th}$ element of $\mathcal{P}_c^s$ to $j^{th}$ element of $\mathcal{P}_c^{s+1}$ is denoted as $\mathcal{E}_{c,s}^{i,j}$. Note that solving weighted multi-partite matching (WMM) is still a hard problem. Hence,  we propose a sequential WMM scheme where sequence $s$ is modeled as a weighted bi-partite matching (WBM) between partitions $\mathcal{P}_c^{s}$ and $\mathcal{P}_c^{s+1}$, $1 \leq s\leq L_c$.  In Fig. \ref{fig:mpartite}, for example, the sequential formation of first cluster $\mathcal{K}_c^{1}$ is highlighted in red colors. Note that proposed  WMM is general enough to apply different clustering purposes by just changing the edge weight design. Each matching sequence of WMM is in the form of rectangular assignment problem (RAP) which is generally solved by Munkres Algorithm with cubic time complexity \cite{Jonker1987}. Accordingly, overall complexity of the proposed sequential WMM based CF is given as $\mathrm{O}(U_c \log U_c + L_c R_c^3)$ whereas exhaustive solution of the original problem has an exponential time complexity, i.e., $\mathrm{O}\left(2^{U_c \times R_c} \right)$.

For an exemplary edge weight design, let us consider edge weight matrix  $H_c \in \mathbb{R}^{R_c \times K_c}$ whose elements are channel gain of corresponding users in the MPG. For example, if $i^{th}$ user of $j^{th}$ partition is UE$_k$ then $H_{i,j}=g_c^k$. Accordingly, we design the edge weight from $i^{th}$ element of $\mathcal{P}_c^s$ to $j^{th}$ element of $\mathcal{P}_c^{s+1}$ as follows
\begin{equation}
\label{eq:weight}
\mathcal{E}_{c,s}^{i,j}= \frac{H_{s,i}}{H_{s+1,j}}
\end{equation}
which makes each WBM sequence favors for new cluster members that maximize the sum of channel gain disparity. Furthermore, Lemma \ref{lem:CF} shows that the previous generic time complexity, $\mathrm{O}(U_c \log U_c + L_c R_c^3)$, can even be reduced to simple sorting and indexing operations of partitioning in \eqref{eq:partit}, i.e., $\mathrm{O}(U_c \log U_c)$. 
\begin{mylemma}[]
\label{lem:CF}
When edge weights of the MPG are merely determined by channel gains as in \eqref{eq:weight}, sequential WMM always forms cluster $\mathcal{K}_c^r$ by $r^{th}$ user of each partition.
\end{mylemma}
\begin{proof}
We first remind that the intra-partition and inter-partition channel gains are in descending order (i.e., $g_c^i \geq g_c^j, \: i,j \in \mathcal{P}_c^s, \: \forall i < j $), and the lowest channel gain within $\mathcal{P}_c^s$ is higher than all channel gains within  $\mathcal{P}_c^{s'}, \: \forall s' > s$. At WBM sequence $s$, therefore, column and row entries of edge weight matrix are monotonically increasing and decreasing, respectively.  That is, $\mathcal{E}_{c, s}^{i,j} \geq \mathcal{E}_{c, s}^{k,j}$ if $i > k$ and  $\mathcal{E}_{c, \ell}^{i,j} > \mathcal{E}_{c, \ell}^{i,k}$  if $j < k$. Accordingly, Munkres Algorithm always matches $i^{th}$ element of partition $s$ with the $i^{th}$ element of partition $s+1$, as it yields the highest weight. Therefore, sequential WMM always forms cluster $i$ by $i^{th}$ member of each partition.  
\end{proof}
\begin{figure*}[b]
\hrule
\begin{align}
\label{eq:Lag}
\nonumber    \mathcal{L}\left(\vect{\omega}, \lambda, \vect{\varphi} ,\vect{\mu}\right) &=\theta B \sum_{i \in \mathcal{K}} \log_2 \left( 1+ \frac{\omega_{i}}{ \sum_{j=1}^{i-1}{\omega_j} +\epsilon_i \sum_{k=i+1}^{K}  \omega_k+ \rho_i} \right) + \lambda \left( \sum_{i \in \mathcal{K}}\omega_i- \varpi \right)  \\
&+ \sum_{i =2}^{K} \varphi_i \left(\omega_i- \sum_{j=1}^{i-1}\omega_j-\Delta_i \right) + \sum_{i =1}^{K} \mu_i \left[ \omega_i- \left(\sum_{j=1}^{i-1} \omega_j  +\epsilon_i \sum_{k=i+1}^{K} \omega_k + \rho_i \right)(q_i-1) \right), 
\end{align}
\begin{align}
\label{eq:Lag_derv_omg2}
\nonumber   \frac{\partial \mathcal{L}}{\partial \omega_i^{\star}}  &=\frac{\theta B  }{\sum_{j=1}^{i}{\omega_j} +\epsilon_i \sum_{k=i+1}^{K}  \omega_k+ \rho_i} - \sum_{m=1}^{i-1} \frac{ \theta B \epsilon_m  \omega_m}{ \left(\sum_{j=1}^{m-1}{\omega_j} +\epsilon_m \sum_{k=m+1}^{K}  \omega_m+ \rho_m \right) \left(\sum_{j=1}^{m}{\omega_j} +\epsilon_m \sum_{k=m+1}^{K}  \omega_k+ \rho_m \right)} \\ 
\nonumber & - \sum_{n=i+1}^{K} \frac{ \theta B \omega_n}{ \left(\sum_{j=1}^{n-1}{\omega_j} +\epsilon_n \sum_{k=i+1}^{K_c^r}  \omega_k+ \rho_c^n \right) \left(\sum_{j=1}^{n}{\omega_j} +\epsilon_n \sum_{k=n+1}^{K_c^r}  \omega_k+ \rho_n \right)} \\
& - \lambda  +\varphi_i  - \sum_{j=i+1}^{K} \varphi_j  +  \mu_i  - \sum_{j =1}^{i-1} \mu_j \epsilon_j \left( q_j -1\right)  - \sum_{k =i+1}^{K}  \mu_k \left( q_k -1\right)  \leq 0,  \omega_{i}^{\star} \geq 0, \\
 \label{eq:Lag_derv_lam}
   \frac{\partial \mathcal{L}}{\partial \lambda^\star} &= \varpi - \sum_{i \in \mathcal{K}_c^r}\omega_i  \geq 0, \text{ if } \lambda^\star  \geq 0, \hspace{5pt} \frac{\partial \mathcal{L}}{\partial \varphi_{i}^{\star}}=\omega_i - \sum_{j=1}^{i-1}\omega_j-\Delta_i \geq 0, \text{ if } \varphi_{i}^{\star} \geq 0,\: i\geq 2, \\
 \label{eq:Lag_derv_mu}
 \frac{\partial \mathcal{L}}{\partial \mu_{i}^{\star}} &=  \omega_i- \left(\sum_{j=1}^{i-1} \omega_j +\epsilon_i \sum_{k=i+1}^{K_c^r} \omega_k + \rho_i \right)(q_i-1) \geq 0,\text{ if } \mu_{i}^{\star} \geq 0.
\end{align}
\end{figure*}

\section{Hierarchically Distributed PBA}
\label{sec:PBA}

In this section, we develop a hierarchical decomposition method in order to obtain a distributed power and bandwidth allocation technique for DL-HetNets. Decomposability of network utility maximization problems provides us with the most appropriate distributed solution methodology and can modularize control and resource allocation in HetNets as it consists of hierarchical network entities \cite{decomp}, i.e., MBSs, SBS, and clusters in our case. Primal decomposition is naturally applicable to master problems where “virtualization” or “slicing” of the resources are carried out by dividing the total resource for each of the cluster competing for the resource \cite{boyd2007notes}. Alternatively, we employ dual decomposition method for the slave problems to obtain optimal primal and dual variables in closed-form which is used by master problems to update power and bandwidth allocation of clusters. 

\subsection{Closed-Form Solutions for Slave Problems}
The slave problems can be formulated as in $\vect{S}$ where we omit the BS and cluster indices for the sake of simplicity. Remarking that $\vect{S}$ depends on bandwidth portion and power fractions given by master problems, its constraints $\mathrm{S}_1$, $\mathrm{S}_2$, and $\mathrm{S}_3$ correspond to $\vect{P}_o$ constraints $\mathrm{C}_{o}^3$, $\mathrm{C}_{o}^6$, and $\mathrm{C}_{o}^7$, respectively. 
\begin{equation*}
\hspace*{0pt}
 \begin{aligned}
 & \hspace*{0pt} \vect{\mathrm{S}}: \underset{\vect{\omega}}{\max}
& & \hspace*{3 pt}  \sum_{i \in \mathcal{K}} \log_2 \left( 1+\gamma_{}^i \right) \\
& \hspace*{0pt} \mbox{$\mathrm{S}_1$: }\hspace*{8pt} \text{s.t.}
&&  \sum_{i \in \mathcal{K}}\omega_i \leq \varpi, \\
 &
  \hspace*{0 pt}\mbox{$\mathrm{S}_2^i$: } & & 0 \leq \omega_i- (q_i-1) \times \\
  && &\left(\sum_{j=1}^{i-1} \omega_j+\epsilon_i \sum_{k=i+1}^{K} \omega_k + \rho_i \right) ,  \textbf{ } \hspace{0 pt}\forall i \in \mathcal{K}\\
&
  \hspace*{0 pt}\mbox{$\mathrm{S}_3^i$: } & & 0 \leq \omega_i- \sum_{j=1}^{i-1}\omega_j-\Delta_i ,  \textbf{ } \hspace{15 pt} \forall i \in \mathcal{K}, i \geq 2
\end{aligned}
\end{equation*}

In order to derive the optimal closed-form expressions for primal and dual variables, we first apply dual decomposition method to the slave problems. Accordingly, Lagrangian function of $\vect{\mathrm{S}}$ can be obtained as in (\ref{eq:Lag}) where $\lambda$, $\varphi_i$, and $\mu_i$ are Lagrange multipliers. Taking derivatives of Lagrangian function w.r.t. $\omega_i$, $\lambda$, $\mu_i$, and $\varphi_i$, Karush-Kuhn-Tucker (KKT)  conditions can be obtained as in (\ref{eq:Lag_derv_omg2})-(\ref{eq:Lag_derv_mu}). Given that some regularity conditions are satisfied, KKT conditions are first-order necessary conditions for a solution in nonlinear programming to be optimal. If linearity constraint qualification is held, i.e., all equality and inequality constraints are affine functions, no other regularity condition is required. This is indeed the case for the slave problems since all constraints are affine functions of power weights. 
\renewcommand{\arraystretch}{0.7}
\begin{table*}[htbp!]
\footnotesize
\centering
\caption{Optimal power allocations along with the corresponding necessary conditions.}
\label{tab:pwr_cond}
\newcolumntype{C}{>{\centering\arraybackslash} m{1cm} }
\begin{tabular}{|l|l|l|}
\hline
\textbf{K  }                                      & \multicolumn{1}{c|}{\textbf{Optimal Power Allocation}} & \multicolumn{1}{c|}{\textbf{Necessary Conditions}} \\ \hline
\multicolumn{1}{|c|}{\multirow{2}{*}{\textbf{2}}} &         \parbox{3cm}{ \begin{align*} \omega_1 &=\frac{\varpi}{q_2}-\frac{\rho_2 (q_2-1)}{q_2}+\psi_1(\vect{\epsilon}) \\ 
\omega_2 &=\frac{\varpi(q_2-1)}{q_2}+\frac{\rho_2(q_2-1)}{q_2} +\psi_2(\vect{\epsilon})
\end{align*}}
                &             \parbox{3cm}{ \begin{align*} & \omega_i- \left(\sum_{j=1}^{i-1} \omega_j +\epsilon_i \sum_{k=i+1}^{K} \omega_k + \rho_i \right)(q_i-1) > 0, \forall i=1.\\
                & \omega_i - \sum_{j=1}^{i-1}\omega_j-\Delta_i > 0, \forall i=2.
\end{align*}}         \\ \cline{2-3} 
\multicolumn{1}{|c|}{}          &           \parbox{3cm}{ \begin{align*} \omega_1 &=\frac{\varpi}{2}-\frac{\Delta_2}{2}+\psi_1(\vect{\epsilon}) \\ 
\omega_2 &=\frac{\varpi}{2}+\frac{\Delta_2}{2} +\psi_2(\vect{\epsilon}) 
\end{align*}}              &              \parbox{3cm}{ \begin{align*} & \omega_i- \left(\sum_{j=1}^{i-1} \omega_j +\epsilon_i \sum_{k=i+1}^{K} \omega_k + \rho_i \right)(q_i-1) > 0, \forall i=1,2.
\end{align*}}             \\ \hline
\multirow{4}{*}{ \textbf{3}}                       &             \parbox{3cm}{ \begin{align*} \omega_1 &=\frac{\varpi}{q_2 q_3}-\frac{\rho_2(q_2-1)}{q_2}-\frac{\rho_3(q_3-1)}{q_2 q_3}+\psi_1(\vect{\epsilon}) \\ \omega_2 &=\frac{\varpi (q_2-1)}{q_2 q_3}+\frac{\rho_2 (q_2-1)}{q_2}-\frac{\rho_3 (q_2-1)(q_3-1)}{q_2 q_3} +\psi_2(\vect{\epsilon}) \\
 \omega_3 &= \frac{\varpi (q_3-1)}{q_3} + \frac{\rho_3(q_3-1)}{q_3}+\psi_3(\vect{\epsilon})
\end{align*}}                &              \parbox{3cm}{ \begin{align*} & \omega_i- \left(\sum_{j=1}^{i-1} \omega_j +\epsilon_i \sum_{k=i+1}^{K} \omega_k + \rho_i \right)(q_i-1) > 0, \forall i=1.\\
                & \omega_i - \sum_{j=1}^{i-1}\omega_j-\Delta_i > 0, \forall i=2,3.
\end{align*}}              \\ \cline{2-3} 
                                         &            \parbox{3cm}{ \begin{align*} \omega_1 &=\frac{\varpi}{2 q_2}-\frac{\rho_2(q_2-1)}{q_2}-\frac{\Delta_3}{2 q_2}+\psi_1(\vect{\epsilon}) \\ 
                                         \omega_2 &=\frac{\varpi (q_2-1)}{ 2 q_2}+\frac{\rho_2 (q_2-1)}{q_2}-\frac{\Delta_3 (q_2-1)}{2 q_2} +\psi_2(\vect{\epsilon}) \\
 \omega_3 &= \frac{\varpi}{2} + \frac{\Delta_3}{2}+\psi_3(\vect{\epsilon})
\end{align*}}                  &          \parbox{3cm}{ \begin{align*} & \omega_i- \left(\sum_{j=1}^{i-1} \omega_j +\epsilon_i \sum_{k=i+1}^{K} \omega_k + \rho_i \right)(q_i-1) > 0, \forall i=1,3.\\
                & \omega_i - \sum_{j=1}^{i-1}\omega_j-\Delta_i > 0, \forall i=2.
\end{align*}}                  \\ \cline{2-3} 
                                         &       \parbox{3cm}{ \begin{align*} \omega_1 &=\frac{\varpi}{2 q_3}-\frac{\Delta_2}{2}-\frac{\rho_3(q_3-1)}{2 q_3}+\psi_1 (\vect{\epsilon})\\ 
                                         \omega_2 &=\frac{\varpi}{2 q_3}+ \frac{\Delta_2}{2}-\frac{\rho_3(q_3-1)}{2 q_3} +\psi_2(\vect{\epsilon}) \\
 \omega_3 &= \frac{\varpi (q_3-1)}{q_3} + \frac{\rho_3(q_3-1)}{q_3} +\psi_3(\vect{\epsilon})
\end{align*}}                            &        \parbox{3cm}{ \begin{align*} & \omega_i- \left(\sum_{j=1}^{i-1} \omega_j +\epsilon_i \sum_{k=i+1}^{K} \omega_k + \rho_i \right)(q_i-1) > 0, \forall i=1,2.\\
                & \omega_i - \sum_{j=1}^{i-1}\omega_j-\Delta_i > 0, \forall i=3.
\end{align*}}                     \\ \cline{2-3} 
                                         &         \parbox{3cm}{ \begin{align*} \omega_1 &=\frac{\varpi}{4}-\frac{\Delta_2}{2}-\frac{\Delta_3}{4}+\psi_1(\vect{\epsilon}) \\ 
                                         \omega_2 &=\frac{\varpi}{4}+ \frac{\Delta_2}{2}-\frac{\Delta_3}{4} +\psi_2(\vect{\epsilon}) \\
 \omega_3 &= \frac{\varpi}{2} + \frac{\Delta_3}{2} +\psi_3(\vect{\epsilon})
\end{align*}}                       &        \parbox{3cm}{ \begin{align*} & \omega_i- \left(\sum_{j=1}^{i-1} \omega_j +\epsilon_i \sum_{k=i+1}^{K} \omega_k + \rho_i \right)(q_i-1) > 0, \forall i=1,2,3.
\end{align*}}               \\ \hline
\end{tabular}
\end{table*}
\begin{figure*}
\begin{mylemma}
\label{lem:cfp}
Given that necessary conditions are satisfied, closed-form optimal power allocation of UE$_1$ and UE$_i$ ($i \geq 2$) are given in \eqref{eq:omega_1} and \eqref{eq:omega_i_1}-\eqref{eq:omega_i_2}, respectively.  
\begin{align}
\label{eq:omega_1}
 \omega_{1}^\star&=\frac{\varpi}{\prod\limits^{K}_{\substack{m=2 \\ m\notin \mathcal{A}_\mu}}q_{m}\prod\limits^{K}_{\substack{n=2 \\ n\in \mathcal{A}_\mu}}2}-\sum\limits^{K}_{\substack{k=2 \\ k\notin \mathcal{A}_\mu}}\frac{(q_{k}-1){\rho_{k}}}{\prod\limits^{k}_{\substack{m=2 \\ m\notin \mathcal{A}_\mu}}q_{m}\prod\limits^{k}_{\substack{n=2 \\ n \in \mathcal{A}_\mu}}2} -\sum\limits^{K}_{\substack{k=2 \\ k\notin \mathcal{A}_\varphi}}\frac{\Delta_{k}}{2\prod\limits^{k-1}_{\substack{m=2 \\ m\notin \mathcal{A}_\mu}}q_{m}\prod\limits^{k-1}_{\substack{n=2 \\ n\in \mathcal{A}_\mu}}2} + \psi_1
\end{align}
\begin{align}
\label{eq:omega_i_1}
 \omega_{i}^\star&=\frac{\varpi}{\prod\limits^{K}_{\substack{m=i \\ m \notin \mathcal{A}_\mu}}q_{m}\prod\limits^{K}_{\substack{n=i \\ n \in \mathcal{A}_\mu}}2}
-\sum\limits^{K}_{\substack{k=i \\ k\notin \mathcal{A}_\mu}}\frac{(q_{k}-1)\rho_{k}}{\prod\limits^{k}_{\substack{m=i \\ m\notin \mathcal{A}_\mu}}q_{m}\prod\limits^{k}_{\substack{n=i \\ n\in \mathcal{A}_\mu}}2}-\sum\limits^{K}_{\substack{k=i \\ k\notin \mathcal{A}_\varphi}}\frac{\Delta_{k}}{2\prod\limits^{j-1}_{\substack{m=i \\ m \notin \mathcal{A}_\mu}}q_{m}\prod\limits^{j-1}_{\substack{n=i \\ n\in \mathcal{A}_\mu}}2}+\Delta_{i} +\psi_{i+1},  i \in \mathcal{A}_\mu. 
\end{align}
\begin{align}
\label{eq:omega_i_2}
  \omega_{i}^\star&=(q_{i}-1) \times \left[\frac{\varpi}{\prod\limits^{K}_{\substack{m=i \\ m \notin \mathcal{A}_\mu}}q_{m}\prod\limits^{K}_{\substack{n=i \\ n \in \mathcal{A}_\mu}}2}
-\sum\limits^{K}_{\substack{k=i \\ k \notin \mathcal{A}_\mu}}\frac{(q_{k}-1){\rho_{k}}}{\prod\limits^{k}_{\substack{m=i \\ m \notin \mathcal{A}_\mu}}q_{m}\prod\limits^{k}_{\substack{n=i \\ n\in \mathcal{A}_\mu}}2} -\sum\limits^{K}_{\substack{k=i \\ k\notin \mathcal{A}_\varphi}}\frac{\Delta_{k}}{2\prod\limits^{k-1}_{\substack{m=i \\ m\notin \mathcal{A}_\mu}}q_{m}\prod\limits^{k-1}_{\substack{n=i \\ n\in \mathcal{A}_\mu}}2}+\rho_{i} \right]+\psi_i,   i \notin \mathcal{A}_\mu
\end{align}
where $\psi_i$ is the term related to SIC imperfections and defined as
\begin{align}
\label{eq:psi}
\nonumber \psi_i (\vect{\epsilon})&= \sum\limits^{K-1}_{\substack{k=i \\ k\in \mathcal{A}_\varphi}}S_{k}\epsilon_{k}(\tilde{q}_k-1)(q_{i}-1) \left \{  \frac{\varpi\left( \prod\limits^{K}_{\substack{m=k+1 \\ m\notin \mathcal{A}_\mu} } q_{m}\prod\limits^{K}_{\substack{n=k+1 \\ n\in \mathcal{A}_2^{}}}2-1 \right)} {\prod\limits^{K}_{\substack{m=2 \\ m\notin \mathcal{A}_\mu} } q_{m}\prod\limits^{K}_{\substack{n=2 \\ n \in \mathcal{A}_\mu}}2}    +\sum\limits^{K}_{\substack{l=k+1 \\l\notin \mathcal{A}_\mu}}\frac{ (q_{l}-1)}{\prod\limits^{l}_{\substack{m=2 \\ m \notin \mathcal{A}_\mu} } q_{m}\prod\limits^{n}_{\substack{n=2 \\ n \in \mathcal{A}_\mu } }2} \right. \\
\nonumber & \left.  +\sum\limits^{K}_{\substack{l=k+1 \\l\in \mathcal{A}_\mu}}\frac{\Delta_{l} }{\prod\limits^{l}_{\substack{m=2 \\ m\notin \mathcal{A}_\mu}}q_{m} \prod\limits^{l}_{\substack{n=2 \\ n \in \mathcal{A}_\mu} } 2} \right \} 
+\sum\limits^{K-1}_{\substack{k=2 \\k\in \mathcal{A}_\varphi}}\sum\limits^{K-k+1}_{\substack{k_{1}=k }}...\sum\limits^{K-1}_{\substack{k_{i-1}=k_{i-2}+1 }}\sum\limits^{K}_{\substack{k_{i}=k_{i-1}+1 }}S_{k1}\epsilon_{k_{1}}\epsilon_{k_{2}}...\epsilon_{k_{i}}\prod\limits^{k_{i}}_{\substack{j=k_{1}}}(q_{j}-1)\\ 
&\times \left \{  \frac{\varpi \left(\prod\limits^{K}_{\substack{m=k_{i}+1 \\ m \notin \mathcal{A}_\mu} } q_{m} \prod\limits^{K}_{\substack{n=k_{i}+1 \\ n\in \mathcal{A}_2^{}}}2-1 \right)} {\prod\limits^{K}_{\substack{m=2 \\ m\notin \mathcal{A}_\mu} } q_{m}\prod\limits^{K}_{\substack{n=2 \\ n\in \mathcal{A}_\mu}}2} +\sum\limits^{K}_{\substack{l=k_{i}+1 \\l\notin \mathcal{A}_\mu}}\frac{ (q_{l}-1)}{\prod\limits^{l}_{\substack{m=2 \\ m\notin \mathcal{A}_\mu} } q_{m} \prod\limits^{l}_{\substack{n=2 \\ n \in \mathcal{A}_\mu } }2}+\sum\limits^{K}_{\substack{l=k_{i}+1 \\l\in \mathcal{A}_\mu}}\frac{\Delta_{l} }{\prod\limits^{l}_{\substack{m=2 \\ m \notin \mathcal{A}_\mu}}q_{m} \prod\limits^{l}_{\substack{n=2 \\ n\in \mathcal{A}_\mu} } 2} \right \}, i \notin \mathcal{A}_\mu.
\end{align}
where $\tilde{q}_k=\begin{cases}
q_{k}, &k\neq i \cr 2, &k=i \cr \end{cases} ,S_{k}=\begin{cases}
-1, &k\neq i \cr 1, &k= i \cr \end{cases} ,\text{ and }S_{k_1}=\begin{cases}
-1, &k_1\neq i \cr 1, &k_1= i\cr \end{cases}.$
\begin{proof}
Please see Appendix \ref{proof:KKT}.
\end{proof}
\end{mylemma}
\begin{mylemma}
\label{lem:opt_lamda}
The optimal value of the slave problems' Lagrange multipliers are derived as
\begin{align}
\label{eq:lag_mul_1}
\lambda_{c,r}^\star&=\begin{cases} \frac{1} { \varpi(t) + \rho^{K} }+\mu_{K}^{\star} 
-\sum_{i =1}^{K-1} \mu_{i}^{\star} \epsilon_i \left( q_i -1\right) \\
 -\sum_{i=1}^{K-1} \frac{ \epsilon_i \omega_i^{\star}}{ \left(\sum_{j=1}^{i-1}{\omega_{j}^{\star}} +\epsilon_i \sum_{k=i+1}^{K}  \omega_{k}^{\star}+ \rho_i \right) \left(\sum_{j=1}^{i}{\omega_{j}^{\star}} +\epsilon_i \sum_{k=i+1}^{K}  \omega_{k}^{\star}+ \rho_{i} \right)}, \text{ if } K \in \mathcal{A}_\mu \\ 
 \frac{1} { \varpi(t)+ \rho_{K_c^r} }+\varphi_{K}^{\star},  \text{ if } K \notin \mathcal{A}_\mu  \end{cases} ,\\
\mu_{2}^{\star}&=\frac{\kappa_{1}^{\star}-\kappa_{2}^{\star}}{q_2}, \: \: \:\mu_{i}^{\star}= \frac{\kappa_{i-1}^{\star}-\kappa_{i}^{\star}+\mu_{i-1}^{\star}(1+\epsilon_{i}(q_{i-1}-1))}{q_i}, i \geq 3, \\
 \label{eq:lag_mul_2}
\varphi_{2}^{\star}&=\frac{\kappa_{1}^{\star}-\kappa_{2}^{\star}}{2}, \: \: \:
\varphi_{i}^{\star}= \frac{\kappa_{i-1}^{\star}-\kappa_{i}^{\star}+\varphi_{i-1}^{\star}}{2},  i\geq3,
\end{align}
where $\kappa_{i}^{\star}$ is obtained by substituting optimal power allocations $\vect{\omega^\star}$ (which are obtained from Lemma \ref{lem:cfp} for a given $\varpi(t)$) into the first three terms of (\ref{eq:Lag_derv_omg2}). 
\begin{proof}
Please see Appendix \ref{proof:KKT}.
\end{proof}
\end{mylemma} 
\end{figure*}

In the slave problem, there exists a total of $2K$ Lagrange multipliers which can be classified into three subsets $\mathcal{A}_1=\{\lambda\}$, $\mathcal{A}_2=\{\mu_i \vert 1 \leq i \leq K\}$, and $\mathcal{A}_3=\{\varphi_i \vert 2 \leq i \leq K \}$. Hence, the solution of each slave problem requires the KKT condition verification of $2^{2K}$ Lagrange multiplier combinations, which is computationally impractical. Fortunately, we only need to check $2^{K-1}$ combinations \cite{Ali2016,chong2013introduction} for the following reasons: Notice that first constraint is always active ($\lambda>0$) at the optimal point as the all available power fraction must be exploited in order to achieve maximum sumrate. However, the power level of UE$_i, \: i \geq 2$, is determined either by PDSC or QoS constraints. That is, UE$_i$ is either active on PDSCs or QoS constraints at the optimal point. 

In order to obtain a closed-form solution, we therefore need to consider the following solution set $\mathcal{A}=\{ \lambda, \mu_i \text{ or } \varphi_i  \vert 2 \leq i \leq K\}$. For instance, if $\mathcal{A}=\{ \lambda,  \mu_2, \varphi_3, \varphi_4 \}$ then the optimal power allocations can be derived from active primal constraints, i.e., $\{ \mathrm{S}_1, \mathrm{S}_2^2, \mathrm{S}_3^3, \mathrm{S}_3^4 \}$, which also requires the satisfaction of corresponding primal KKT conditions, i.e., $\{ \mathrm{S}_2^1, \mathrm{S}_3^2, \mathrm{S}_2^3, \mathrm{S}_2^4 \}$. Table II shows the optimal power allocations and corresponding KKT conditions for cluster sizes $K=2$ and $K=3$. In Table II, the terms $\psi_i (\vect{\epsilon}), \forall i,$ are due to the residual interference and mainly characterized by the FEF. On the other hand, the terms including $\Delta_i$ are because of the PDSCs and chiefly determined by $p_\Delta$. For each cluster size, the left column provides the power allocations cluster members for each of $2^{K-1}$ cases whereas the right column indicates the set of KKT conditions to be satisfied by the corresponding power allocations on the same row of the left column. For each of $2^{K-1}$ cases, power allocations are first computed by expressions on the left column, then substituted into the corresponding equations on the right column to check if the corresponding KKT conditions are satisfied. Thereafter, optimal power allocations are determined by the case which gives the highest sumrate among the cases who satisfy the KKT conditions. Therefore, the worst case time complexity of the proposed solution in Lemma II is given as $\mathrm{O}(2^{K-1}+K\log K) $ where the first term is the cost of calculating and checking $2^{K-1}$ cases and the second term is the cost of sorting and selecting the best case. Since the complexity of the first term dominates that of the second, overall complexity can be approximated by $\mathrm{O}(2^{K-1})$. Generalizing the Table \ref{tab:pwr_cond}, Lemma \ref{lem:cfp} provides the closed-form expression for optimal power allocations for an imperfect NOMA cluster of size $K$, where $\mathcal{A}_\mu=\mathcal{A} \setminus \mathcal{A}_2$, $\mathcal{A}_\varphi=\mathcal{A} \setminus \mathcal{A}_3$. 

\subsection{Secondary Master Problems}
Secondary master problems are responsible for determining the power allowances of each cluster $\varpi_c^r$, which is iteratively updated according to the feedback received from slave problems. Based on the optimal power allocations obtained from Lemma \ref{lem:cfp}, $\vect{\omega_{c,r}^\star}$, secondary master problems can be formulated at iteration $t$ as in $\vect{\mathrm{M}_2^c (t)}$, $0 \leq c \leq S$.
\begin{equation*}
\hspace*{0pt}
 \begin{aligned}
 & \hspace*{0pt} \vect{\mathrm{M}_2^c(t)}: \underset{ \vect{0} \preceq \vect{\varpi_c(t)}\preceq \vect{1}}{\max}
& & \hspace*{3 pt}  \theta_c^r B \sum_{i \in \mathcal{K}_c^r} \log_2 \left( 1+\gamma_{c,r}^i(\vect{\omega_{c,r}^\star(t-1)}) \right) \\
& \hspace*{10pt} \mbox{ }\hspace*{40pt} \text{s.t.}
&&  \sum_{r \in \mathcal{R}_c} \varpi_c^r(t) \leq 1, \forall c
\end{aligned}
\end{equation*}
which handles the power fraction allocated for clusters and executed by BSs in a parallel and independent fashion. The only constraint of $\vect{\mathrm{M}_2^c}$ corresponds to $\mathrm{C}_o^3$ of $\vect{\mathrm{P}_\mathrm{o}}$. By exploiting the primal decomposition method \cite{decomp},  BS$_c$ can update the power fractions of clusters  based on the following subgradient method \cite{boyd2007notes}
\begin{equation}
\label{eq:grad}
\varpi_c^r(t+1)=\left[ \varpi_c^r(t)+ \nu \lambda_{c,r}^{\star} (\varpi_c^r(t))\right]_{\Pi}
\end{equation}
where $t$ is time index, $\nu$ is a positive step size, $[\cdot]_{\Pi}$ denotes the projection onto the feasible set $\Pi \triangleq \{ \sum_{r \in \mathcal{R}_c} \varpi_c^r(t) \leq 1, 0 \leq \varpi_c^r(t) \leq 1\}$, and  $\lambda_{c,r}^{\star} \left( \varpi_c^r(t) \right)$ is the subgradient given by the optimal Lagrange multiplier associated to the first constraint of the slave problem, which is given in Lemma \ref{lem:opt_lamda}.
\subsection{Primary Master Problem}
Contingent upon the optimal power allocations, the primary master problem determines the bandwidth allocation at iteration $k$ as in $\vect{\mathrm{M}_1(k)}$ 
\begin{equation*}
 \begin{aligned}
 & \hspace*{0pt} \vect{\mathrm{M}_1(k)}: \underset{\vect{0} \preceq\vect{\theta(s)} \preceq \vect{1}}{\max}
& & \hspace*{-5 pt}  B \sum_{c, r \in \mathcal{R}_c} \theta_c^r(k) \sum_{i \in \mathcal{K}_c^r} \log_2 \left( 1+\gamma_{c,r}^i( \theta_c^r(k)) \right) \\
& \hspace*{0pt} \mbox{ }\hspace*{50pt} \text{s.t.}
&& \hspace*{-10pt}  \sum_{c,r \in \mathcal{R}_c} \theta_c^r(k) \leq 1 \\
& & & 0 \leq \omega_{c,r}^{i,\star}- (q_{c,r}^i-1) \times \\
&&& \hspace*{-10pt}  \left(\sum_{j=1}^{i-1} \omega_{c,r}^{j,\star}+\epsilon_i \sum_{k=i+1}^{K_c^r} \omega_{c,r}^{k,\star} + \rho_{c,r}^i \right) , \textbf{ } \forall c,r.
\end{aligned}
\end{equation*}
which is a highly non-convex problem as the objective is a non-convex function of $\theta_c^r$, i.e., $\rho_{c,r}^i=N_0B\theta_c^r(k)(P_c g_c^i)^{-1}$. The first and second constraints of $\vect{\mathrm{M}_1(k)}$ corresponds to $\mathrm{C}_o^5$ and $ \mathrm{C}_o^6$ of $\vect{P}_o$, respectively. Accordingly, we relax $\vect{\mathrm{M}_1(k)}$ into a convex problem by employing the achieved spectral efficiency in the previous iteration. That is, the objective function is changed to $ B \sum_{c, r \in \mathcal{R}_c} \theta_c^r(k) U_c^r\left[\theta_c^r(k-1)\right]$ where $U_c^r\left[\theta_c^r(k-1)\right] =\sum_{i \in \mathcal{K}_c^r} \log_2 \left( 1+\gamma_{c,r}^i( \theta_c^r(k-1)) \right)$ is the utility obtained by clusters in the previous iteration. We refer interested readers to Appendix \ref{proof:conv} for the convexity analysis of the primary master problem and its relaxation. It is worth remarking that $\vect{\mathrm{M}_1(k)}$ only requires the reports of utility achieved by clusters and there is no need to know any other information, e.g., clusters, channel gains, power allocations, etc.

\begin{algorithm}
 \caption{\small Hierarchical Distributed Power-Bandwidth Allocation}
  \label{alg:HD}
\begin{algorithmic}[1]
\small
 \renewcommand{\algorithmicrequire}{\textbf{Input:}}
 \renewcommand{\algorithmicensure}{\textbf{Output:}}
\REQUIRE Channel gains
 \STATE $(k,t) \leftarrow 0$
  \STATE $\vect{\alpha_c} \leftarrow$ BS$_c$ forms its clusters using WMM method, $\forall c$.
    \STATE $\vect{\theta(k)} \leftarrow$   Initialize the bandwidth allocations, $\forall c, r$. 
  \WHILE {$k \in \mathcal{T}_1$}
   \WHILE {$t \in \mathcal{T}_2$}
  \STATE $\vect{\omega_{c,r}(t)} \leftarrow$ Check power levels \& conditions, $\forall c, r$.\STATE $\vect{\omega_{c,r}^{\star} (t)} \leftarrow$ Select the cases with max. sumrate, $\forall c, r$.\STATE $\vect{\lambda^\star(t)} \leftarrow$ Calculate optimal Lagrange multipliers.\STATE $\vect{\varpi_{c}^r (t+1)} \leftarrow$ Update power fractions by (\ref{eq:grad}), $\forall c, r.$\STATE  $t \leftarrow t+1$ 
\ENDWHILE
\STATE Receive utilities $\vect{U_c(k)}$ from all BSs.
\STATE $\vect{\theta(k+1)}  \leftarrow$ Update bandwidths by $\vect{\mathrm{M}_1(k)}$.
\STATE  $k \leftarrow k+1$
\ENDWHILE
\hspace{-16pt}\RETURN  Power and bandwidth allocations
 \end{algorithmic}
 \end{algorithm}
 
Proposed distributed framework is summarized in Algorithm \ref{alg:HD} which is certainly a visualization of Fig. \ref{fig:distributed}.  In Algorithm \ref{alg:HD}, BSs are only required to know channel gains of their own UEs. Following the initialization of the cluster bandwidths, which are independently executed by each BS in the order of $\mathrm{O}(U_c \log U_c)$, bandwidth and power allowances of clusters are updated by outer and inner while loops, respectively. The inner loop is processed by BSs in a parallel fashion, where optimal powers and Lagrange multipliers are first calculated by (\ref{eq:omega_1})-(\ref{eq:omega_i_2}) and (\ref{eq:lag_mul_1})-(\ref{eq:lag_mul_2}), respectively. Notice that calculations of (\ref{eq:omega_1}) and (\ref{eq:omega_i_2}) is the most time consuming part of the while loops and it can be given for BS$_c$ as $\mathrm{O}\left(2^{\sum_r (K_c^r-1)}\right)$. Thereafter, power fractions for each cluster are updated as per (\ref{eq:grad}) until a termination constraint satisfied. For outer loop, the MBS first receive utilities of clusters from each BS, updates the cluster bandwidths as per primary master problem, and disseminates results to the BSs to receive utilities of the most recent update. Please note that only physical message passing takes place between smallcells and the MBS to report obtained utilities, which is in order of the total number of clusters and decreases with higher cluster sizes. Therefore, proposed distributed method has a low communication overhead since it only requires the broadcasting of $R$ bandwidth proportion and corresponding cluster utility feedback to update the bandwidth allocation. Notice that message passing between secondary master problem and its slaves does not cause any communication overhead as they are internal operations of BS processors.

Although Algorithm \ref{alg:HD} optimizes both power and bandwidth to satisfy QoS requirements with the maximum possible network sumrate, it is possible to encounter some infeasible NOMA scenarios. Since NOMA is already considered as an alternative transmission scheme underlying the existing OMA systems (e.g., Multiuser Superposition Transmission (MUST) \cite{Standard15}), hybridization of OMA and NOMA schemes would be quite useful to handle such scenarios by switching members of an infeasible NOMA cluster to OMA mode. 

\section{Numerical Results and Analysis}
\label{sec:res}
For the simulations, we consider $U$ UEs and $S$ SBSs uniformly distributed over a cell area of $500\:m \times 500\:m$ MBS. QoS requirements of UEs are randomly determined with a mean of 1.5 Mbps. UEs are assumed to be equipped with a single antenna and single SIC receiver. All results are obtained by averaging over 100 network scenarios. Unless it is stated explicitly otherwise, we use the default simulation parameters given in Table \ref{tab:parameters}. The composite channel gain, $g_c^i$, between BS$_c$ and UE$_i$ is given as
\begin{equation}
\label{eq:gij}
g_c^i=A_c^i \delta_{c,i}^{-\eta_c^i} 10^{\xi_c^i/10} \E\{ | \tilde{g}_c^i|^2 \}
\end{equation}
where  $A_c^i $ is a constant related to antenna parameters, $\delta_{c,i}$ is the distance between the nodes, $\eta_c^i$ is the path loss exponent, $10^{\xi_c^i/10}$ represents the log-normally distributed shadowing, $\xi_c^i $ is a normal random variable representing the variation in received power with a variance of $\varrho_c^i $, i.e., $\xi_c^i \sim \mathcal{N}(0,\varrho_c^i )$, $\tilde{g}_c^i$ is the complex channel fading coefficient, $\E\{ \cdot \}$ is the expectation to average small scale fading out, and $ \E\{ | \tilde{g}_c^i|^2 \}$ is assumed to be unity.  

\subsection{Benchmark Comparison of WMM Based Clustering }
As the clustering has a significant impact on NOMA performance, we start with the comparison of the proposed sequential WMM based clustering and the following benchmark 
\begin{equation}
\underset{\vect{\alpha_c}}{\max} \: \sum_{r=1}^{R_c} \sum_{i=1}^{K_c^r} \log(1+\gamma_{c,r}^i) \text{ s.t. } \sum_{r=1}^{R_c}\alpha_{c,r}^i \leq K_c^r, \sum_{i=1}^{K_c^r} \alpha_{c,r}^i =1, \forall c.
\end{equation} 
which is handled by each BS independently. As this subsection merely focuses on the CF performance, clusters are assumed to have equal power and bandwidth allowance such that available cluster power is allocated to UEs inversely proportional to their channel gains. Benchmark problem is obtained by Solving Constraint Integer Programs (SCIP) solver \cite{MaherFischerGallyetal.2017} using OPTI toolbox developed for Matlab \cite{CW12a}. The ratio of benchmark to proposed sequential WMM based CF is presented in Fig. \ref{fig:CF1}. We omit the cluster size of $2$ as the optimal clustering can simply be obtained by bi-partite matching. The time complexity of the benchmark is illustrated in Fig. \ref{fig:CF2} where elapsed time ranges between 30 to 90 minutes for cluster sizes 3 to 7, respectively. On the other hand, the time complexity of the proposed CF is in the magnitude order of milliseconds. Fig. \ref{fig:CF} obviously show that proposed WMM based method provides very fast yet high CF performance. 

\begin{table}[t!]
\centering
\caption{Table of Parameters}
\label{tab:parameters}
\scriptsize 
\resizebox{0.48\textwidth}{!}{%
\begin{tabular}{|l|l|l|l|l|l|}
\hline
Par.         & Value       & Par.          & Value           & Par.          & Value         \\ \hline
$\eta_c^u$   & $3.76$   &  $\sigma_c^u$ & $10$ $dB$      & $\beta$   & $0.3$       \\ \hline
   $N_0$         & $-174$ $dBm$ &  $\epsilon_u$      & $10^{-5}$   & $p_{\Delta}$      & $-90$ $dBm $        \\ \hline
$K_c^r$        & $5$  & $P_m$ & $46$ $dBm$   &  $P_s$  & $30$ $dBm$    \\ \hline
$B$      & $20$ $MHz$  & $S$           & $10$             & $ U$ & $100$     \\ \hline
\end{tabular}%
}
\end{table}

\begin{figure}[t!]
    \centering
        \begin{subfigure}[b]{0.24\textwidth}
\includegraphics[width=\textwidth]{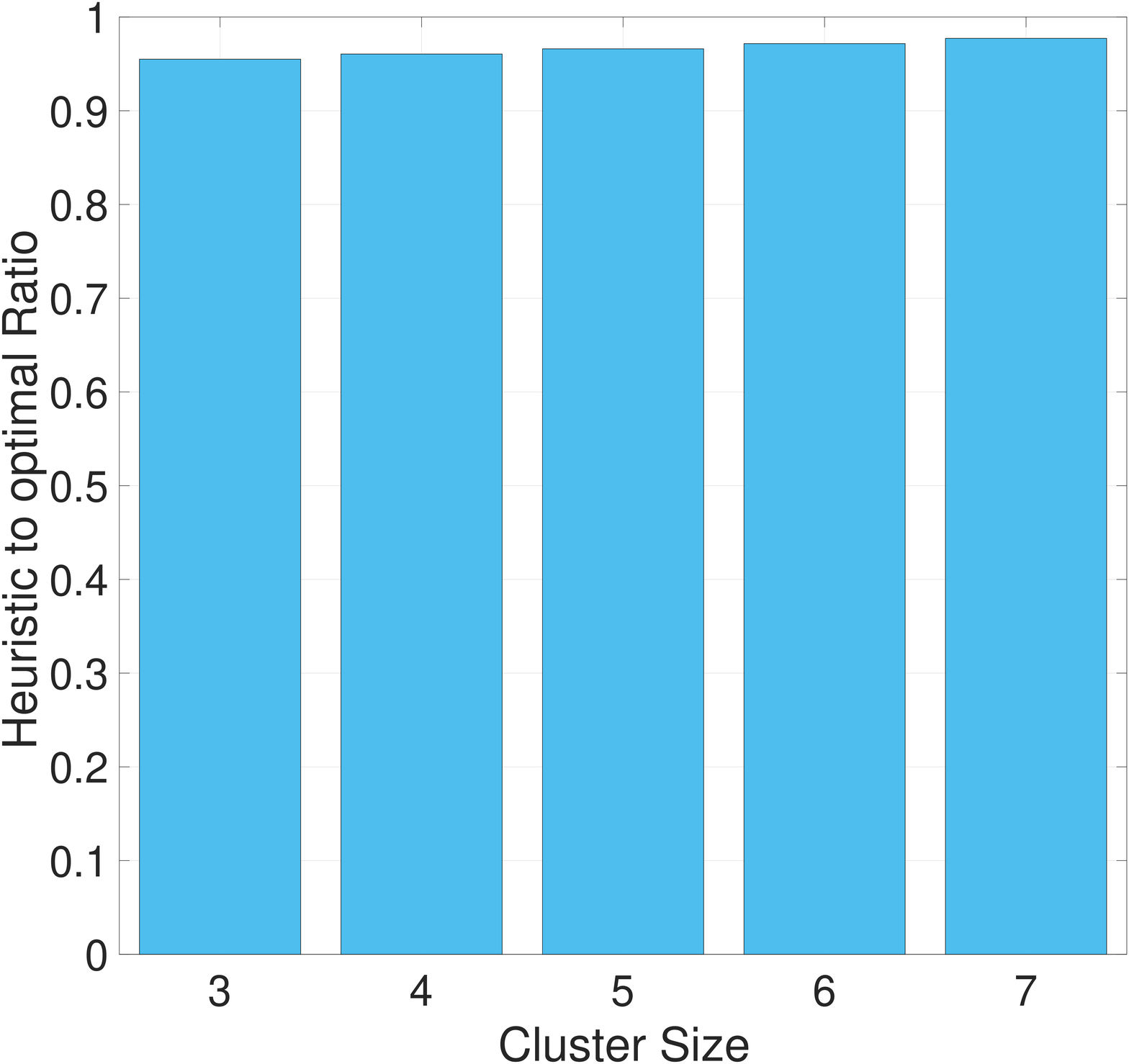}  
        \caption{Performance comparison.}
 \label{fig:CF1}
    \end{subfigure}
    \begin{subfigure}[b]{0.24\textwidth}
\includegraphics[width=\textwidth]{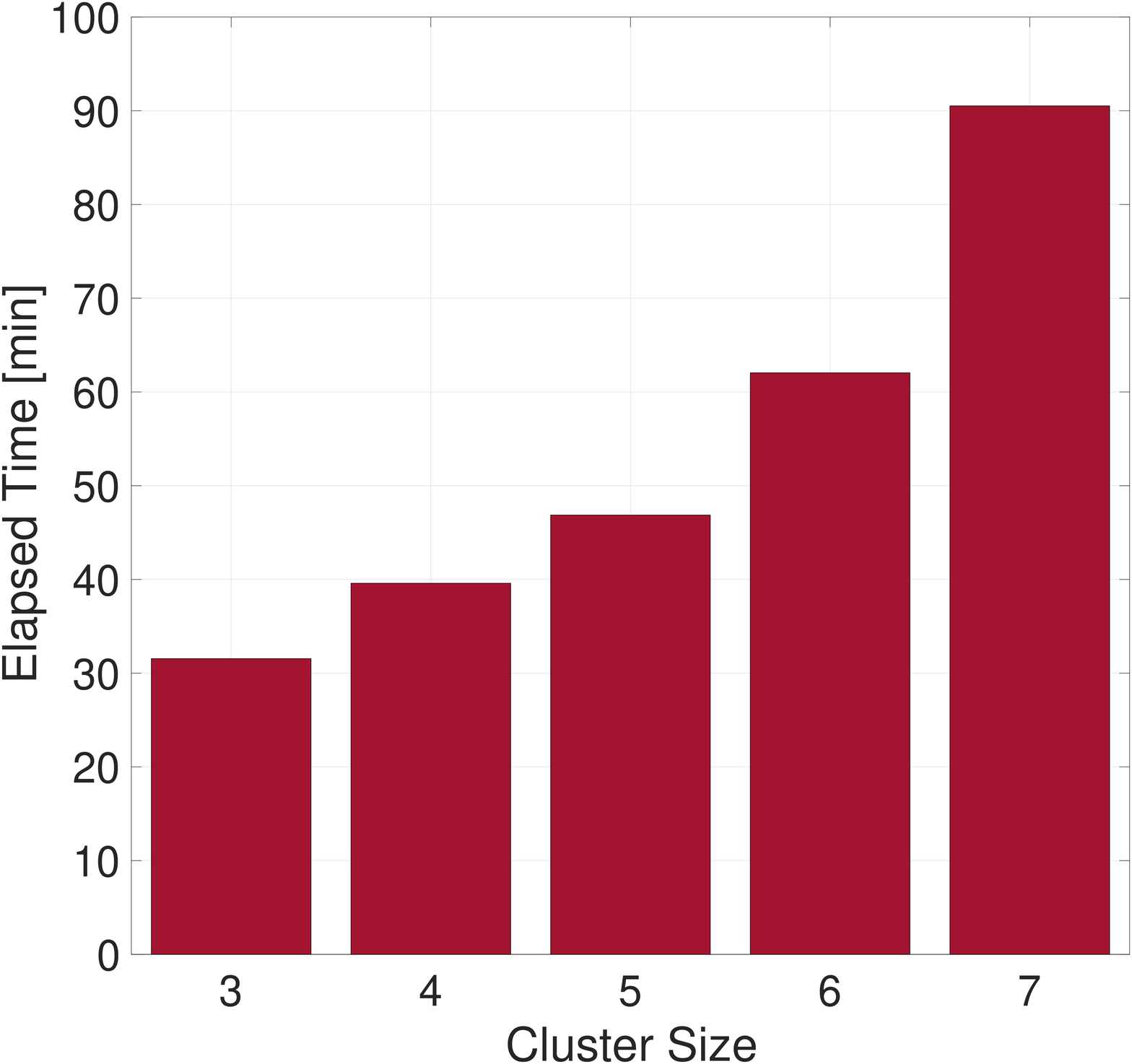}
  \caption{Time complexity comparison.}
  \label{fig:CF2}
    \end{subfigure}   
      \caption{Benchmark comparison of WMM based CF.}
    \label{fig:CF}
\end{figure}

\subsection{Impact of ICRI and ICI on NOMA Performance}

In order to investigate optimal power control behavior of an imperfect NOMA cluster, let us consider a sample cluster of size four whose members are located at 50 m, 100m, 200 m, and 400 m far away from the serving BS. Fig. \ref{fig:comp} illustrates the proposed optimal power allocations, Intra-cluster interference, and user rates in blue, red, and orange colors, respectively. While filled green markers are optimal power allocations of perfect NOMA ($\epsilon=0$) obtained based on \cite{Ali2016}, purple markers are power levels calculated by the proposed solution. As it is already pointed out in the paper, the proposed solution reduces exactly to the perfect NOMA case when $\epsilon=0$. The influence of the FEF on the residual interference can be observed from red lines which monotonically increase with $\epsilon$. On the other hand, the power weight of the last (weakest) user reduces as $\epsilon$ increases whereas those of stronger users keep escalating with $\epsilon$. Reminding that the first (strongest) user achieves the highest rate by canceling the all the interference in the perfect case, boosting $\epsilon$ yields more residual interference mostly from the weakest user. To compensate the resulting performance degradation, first user's SINR is improved by incrementing $\omega_1$ and diminishing impact of the most dominant interferer, i.e., $\omega_4$. A similar discussion also applies to other power levels, i.e., $\omega_2$ and $\omega_3$. Fig. \ref{fig:comp}  also demonstrates how NOMA can deliver a fair service while increasing the overall sum rate at the same time. As can be observed from orange lines, the strongest user obtains the highest rate while other users enjoy close rates which are around 30-40\% less than the first user's rate. This gap further reduces as $\epsilon$ increases and all users are provided by almost the same performance starting from $\epsilon=10^{-8}$. That is, performance loss caused by the SIC imperfection is compensated by extra rate offered to the strongest user.

   \begin{figure}[t!]
 \centering
 \includegraphics[width=0.48 \textwidth]{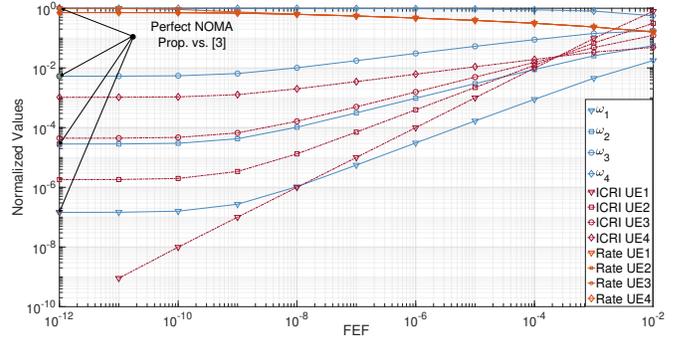}
 \caption{Impact of residual interference on power allocations and rates.}
  \label{fig:comp}
 \end{figure}
 \begin{figure}[t!]
    \centering
        \includegraphics[width=0.48 \textwidth]{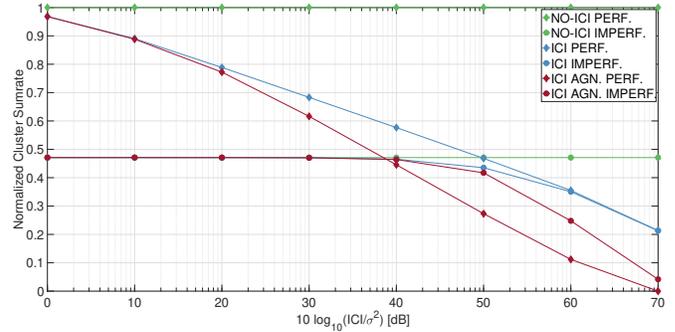}
        \caption{Impact of the ICI on the perfect and imperfect NOMA cases.}
        \label{fig:ICI}
\end{figure} 

Although the considered system model mitigates the intra-cell interference within a macro cell coverage by SIC and dedicating the entire available DL spectrum to clusters of the MBS and SBS within its coverage, neighboring BSs may still cause ICI which can be quite significant and further deteriorate the cluster performance severely. Noting that the ICI can possibly not be canceled by SIC receivers, it is expected to have a considerable influence especially on the cell-edge users, i.e., the weakest user in a cluster. Using the same setup in Fig. \ref{fig:comp}, now let us investigate the impact of ICI on the cluster performance. Fig. \ref{fig:ICI} depicts the normalized cluster sumrate against $10 \log_{10}(\frac{I_{ici}}{\sigma^2})$ that quantifies the ICI by multiples of the receiver noise power. The negative impacts of ICI on the optimal power allocation schemes can be understood by contrasting the blue and green colored ICI and NO-ICI curves, respectively. It is obvious that perfect NOMA ($\epsilon=0$) case experience a drastic performance degradation starting from very low ICI values since $I_{ici}$ is not negligible in comparison with the ICRI. On the other hand, the ICI influence on the imperfect NOMA case ($\epsilon=10^{-5}$) is distinguishable only after 40 dB (i.e., $I_{ici}$ is $10^4$ times of the receiver noise) which is exactly where the ICI starts becoming significant compared to the ICRI. More importantly, the agnostic case curves clearly illustrate that ignoring the $I_{ici}$ term in the optimal power allocations can cause even worse performance degradation. Again, the level of performance reduction in the perfect NOMA is more than that in the imperfect NOMA since the power levels obtained based on the wrong assumption of no ICI is much more different than actual ones in the former case. Fig. \ref{fig:ICI} clearly shows that the ICI must be taken into account especially when it becomes considerable by comparison with the ICRI and receiver noise. Therefore, it is necessary to suppress the dominance of the ICI for a desirable network performance.

\begin{figure}[!t]
    \centering
        \begin{subfigure}[b]{0.5\textwidth}
        \includegraphics[width=1\textwidth]{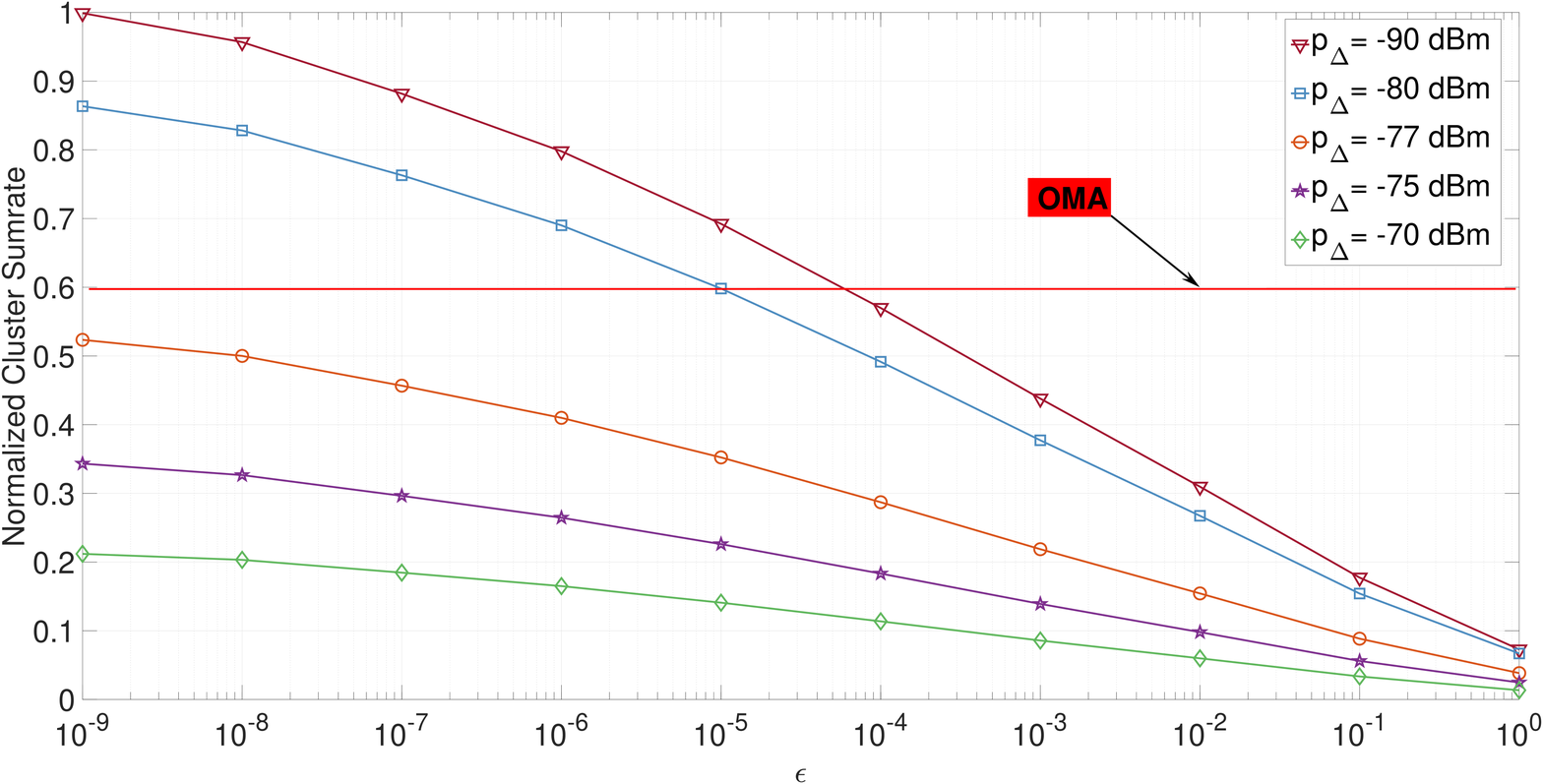}
        \caption{}
        \label{fig:epsilon}
    \end{subfigure}
    \begin{subfigure}[b]{0.5 \textwidth}
        \includegraphics[width=1\textwidth]{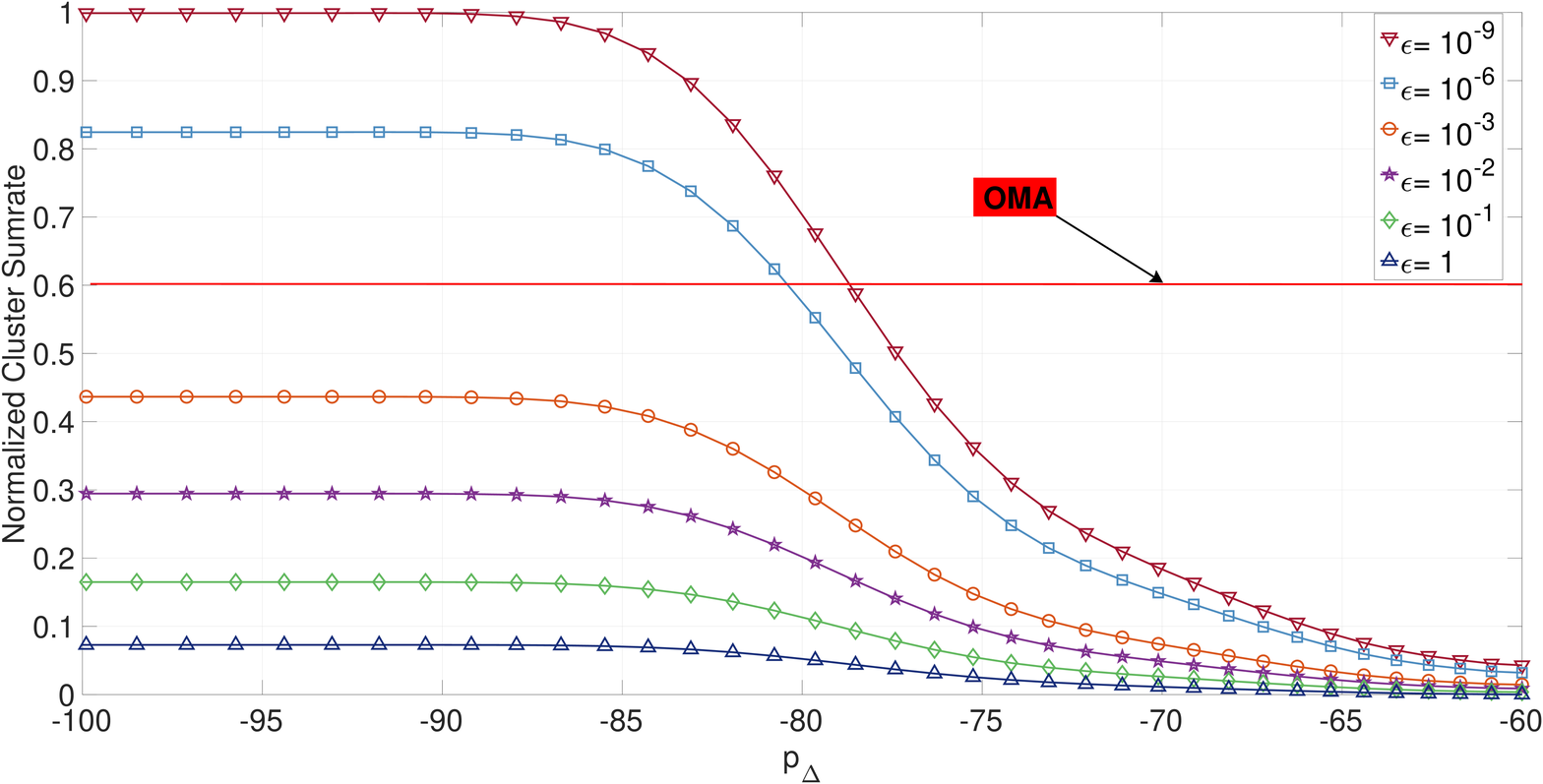}
        \caption{}
        \label{fig:delta}
    \end{subfigure}   
    \caption{Normalized sumrate vs. (a) FEF, $\epsilon$, and (b) sensitivity, $p_\Delta$.}
\end{figure}

\begin{figure}[!t]
    \centering
        \begin{subfigure}[b]{0.5\textwidth}
        \includegraphics[width=1\textwidth]{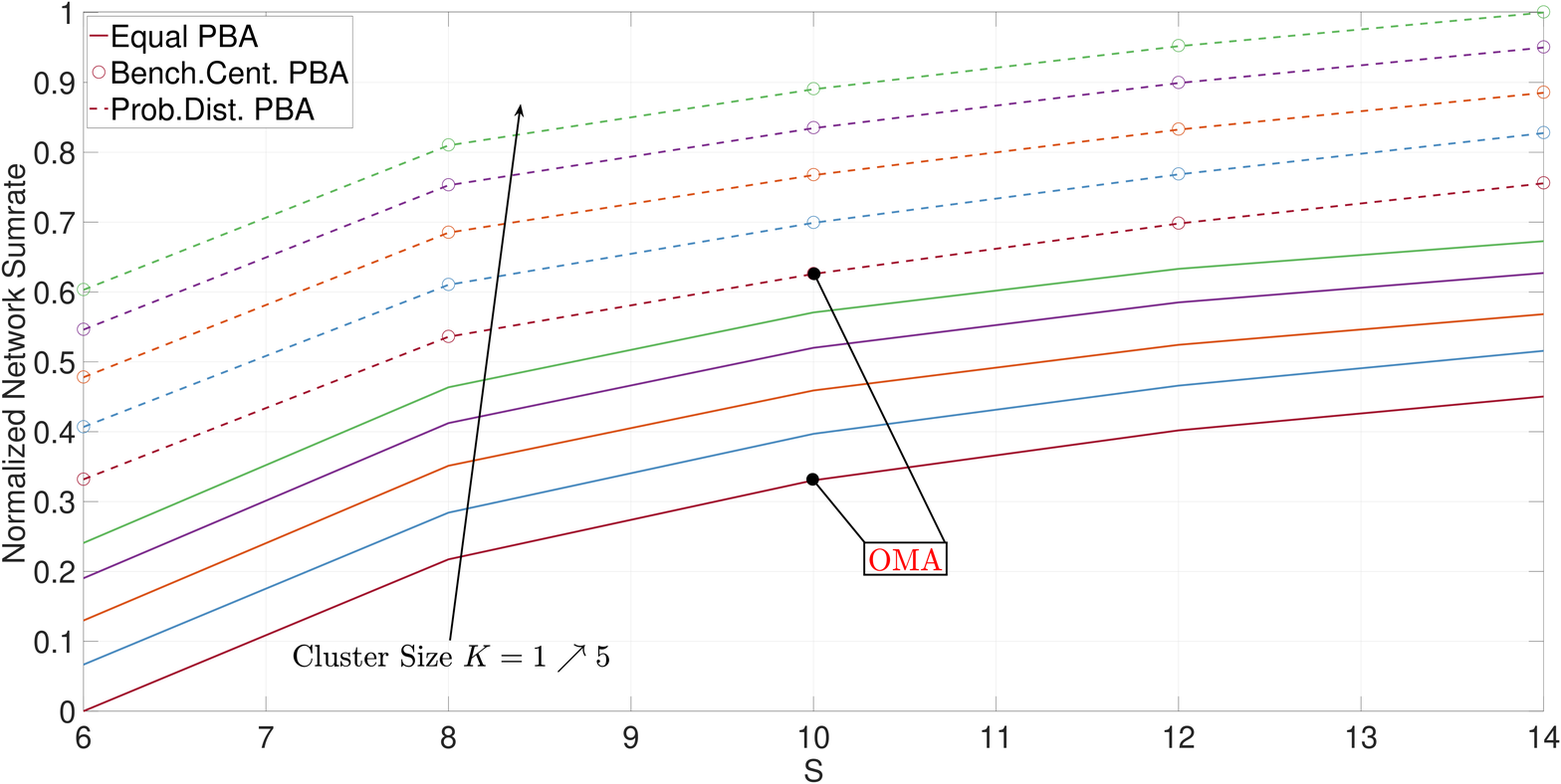}
        \caption{}
        \label{fig:lamdaS}
    \end{subfigure}
    \begin{subfigure}[b]{0.5 \textwidth}
        \includegraphics[width=1\textwidth]{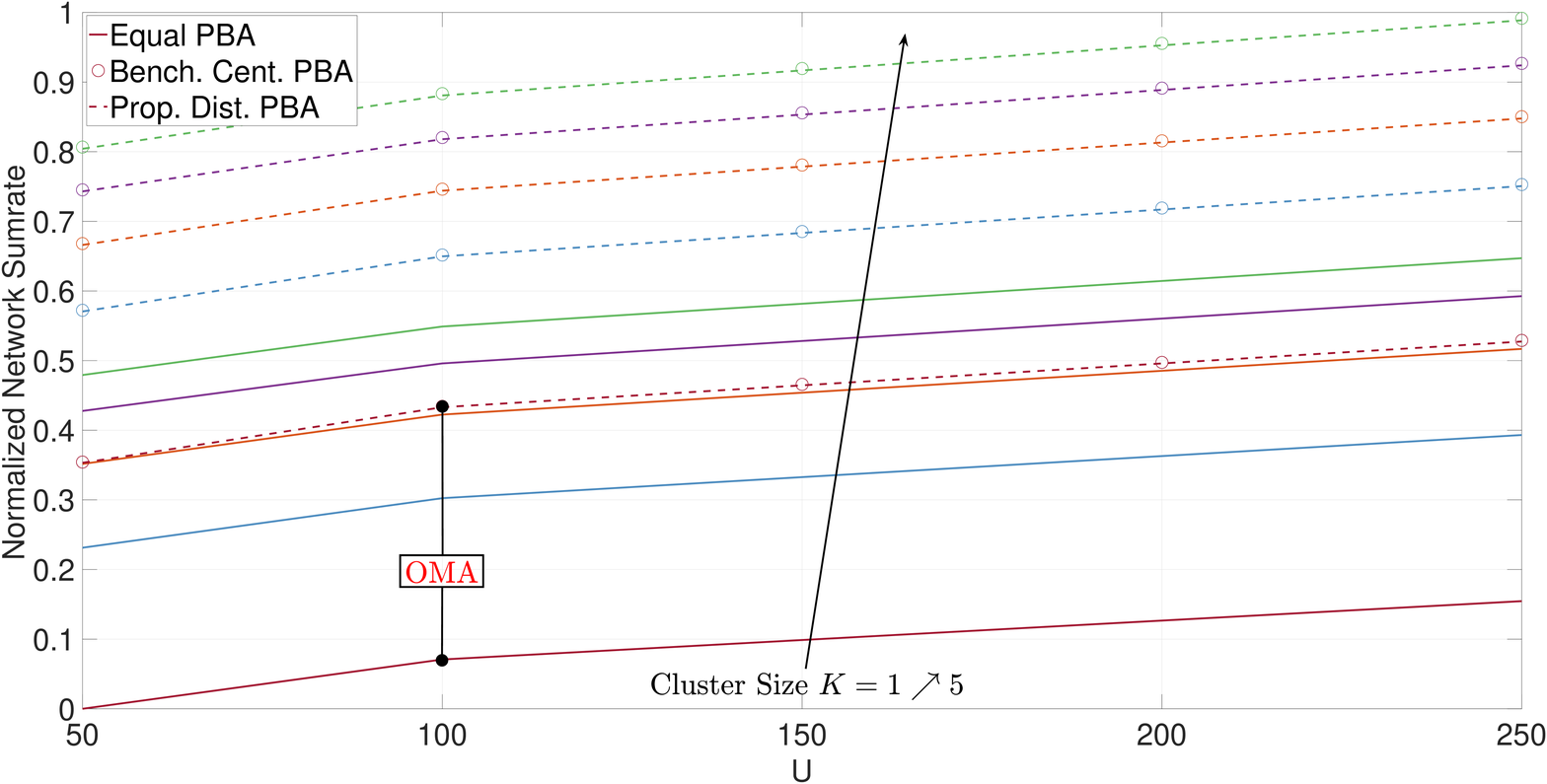}
        \caption{}
        \label{fig:lamdaU}
    \end{subfigure}   
        \begin{subfigure}[b]{0.5 \textwidth}
        \includegraphics[width=1\textwidth]{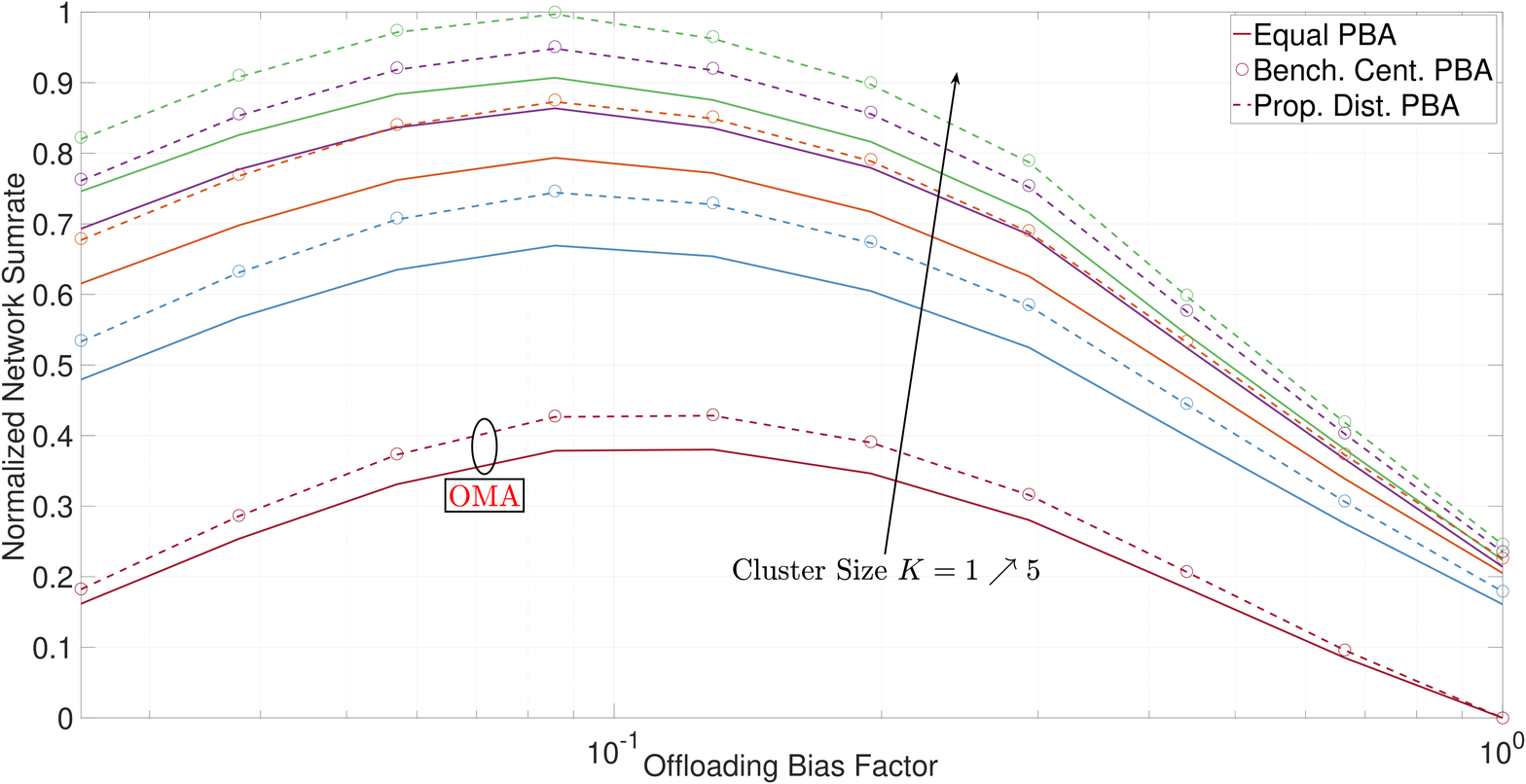}
        \caption{}
        \label{fig:beta}
    \end{subfigure}   
    \caption{Normalized network sumrate vs. a) number of SBSs, b) number of UEs, and c) the offloading bias factor.}
\end{figure}

Fig. \ref{fig:epsilon} and  Fig. \ref{fig:delta}  demonstrate the impacts of residual interference and receiver sensitivity on the normalized sumrate performance, respectively. For the sake of comparison, optimal OMA performance is also drawn with a red colored solid line. While the system performance monotonically decreases with the error factor due to the increasing residual interference, it starts to degrade dramatically for sensitivities higher than $-90$ dBm. Therefore, the efficiency of SIC receivers is a key parameter to determine the NOMA gain. As $\epsilon \rightarrow 1$ NOMA starts to behave like a co-channel OMA scheme without any interference control. The main reason behind the performance degradation caused by the sensitivity is that receiver sensitivity enforces the optimal power allocations to ensure a certain power gap, which starts deteriorating the overall performance after a certain sensitivity level.  

Fortunately, receiver sensitivities of GSM, HSPA, and LTE uplink budgets are given as $-114$ dBm,  $-123.4$ dBm, and $-123.4$ dBm, respectively. Similarly, receiver sensitivities of GSM, HSPA, and LTE downlink budgets are given as $-104$ dBm,  $-106.4$ dBm, and $-106.4$ dBm, respectively \cite{holma2009lte}. Therefore, sensitivity is not a challenge for today's cellular devices as their receiver sensitivity is below $-100$ dBm which does not have a significant effect on the NOMA performance as shown in Fig. \ref{fig:delta}. However, since NOMA has a potential to simultaneously serve massive machine type communications, one still need to take the sensitivity into account to design NOMA schemes for low-cost devices with high receiver sensitivities. As can be deduced from Fig. \ref{fig:epsilon} and Fig. \ref{fig:delta}, a realistic NOMA power control scheme must consider sensitivity and residual interference as they primarily determine the performance enhancement obtained by NOMA scheme.

\subsection{Benchmark Comparisons of Proposed Distributed PBA }
In this section, we compare proposed distributed PBA with two benchmarks: 1) \textit{Equal PBA} as in \cite{Ali2016} where authors consider equally allocated cluster powers and 2)  \textit{Centralized PBA} is obtained by geometric programming framework based on successive convex approximation method \cite{celik2017DLNOMA}. As an alternative method for the centralized scheme, one can also consider SCALE which also involves a sequence of convex approximations \cite{SCALE}. Network sumrate performance is investigated under different cluster sizes with respect to essential network parameters including $S$, $U$, and offloading bias factor. Please note that unit cluster size ($K=1$) corresponds to the power and bandwidth optimal OMA scheme where each UE has its own channel. Centralized PBA is solved by CVX which is a Matlab software for disciplined convex programming \cite{cvx}. For the distributed PBA, step sizes are set to $\nu=0.005$ which yields a convergence within an average of 200 iterations. On the other hand, the outer loop is converged within an average of 10 iterations, which is desirable as the outer loop involves message passing between SBSs and MBS. As expected, the centralized scheme delivers a better performance than the distributed one. However, since the performance gap between the centralized and distributed schemes are not significant with respect to actual sumrate values, this gap became harder to be recognized after the normalization process.     

Normalized network sumrate performance of PBA schemes with respect to increasing smallcell density and cluster sizes are presented in Fig. \ref{fig:lamdaS}. Apparently, the network performance monotonically improves with increasing BS density and cluster sizes, which can be explained as follows: While a higher number of BSs provides higher SINRs since UEs are associated with closer BS, i.e., higher channel gains, increasing cluster size improves the spectral efficiency. Please note that incrementing cluster size linearly increases the time complexity of SIC process and incur extra processing delay \cite{Andrews2005SIC}. Both centralized and distributed PBA schemes obviously outperform the equal PBA scheme. Moreover, the distributed solution gives almost the same performance with the centralized approach with its low communication overhead. During the simulations, we observe that optimal power allocation of perfect NOMA distributes available BS power to its clusters almost uniformly. However, this is not the case for imperfect NOMA since residual interference starts playing a significant role in behavior as $\epsilon \rightarrow 1$. Nevertheless, distributed PBA still yields notable improvement due to the iterative bandwidth allocation.

Network performance of PBA schemes with respect to increasing UE density and cluster sizes are shown in Fig. \ref{fig:lamdaU}. Similar to Fig. \ref{fig:lamdaS}, the network performance monotonically enhances with increasing UE density and cluster sizes. It can be deduced from Fig. \ref{fig:lamdaU} that NOMA scheme can serve a large number of UEs with better network sumrate thanks to its inherent fairness and spectral efficiency features. Finally, Fig. \ref{fig:beta} depicts the system performance with respect to different cluster sizes and offloading bias factor. NOMA gain is higher for smaller bias factors as offloading more UEs to smallcells give a better chance to form clusters with a higher channel gain disparity. Notice that improvement over equal PBA and smaller cluster sizes are also becoming gradual as $\beta \rightarrow 1$ as $\beta \rightarrow 1$ associates most of the UEs with the MBS and they compete for the same power source. Finally, sumrate performance first increases and then decreases as $\beta$ increases, which is resulted from the variation in UEs’ SINR due to the UE association. That is, clustering and UE association pairs provide the highest SINR case when $\beta \simeq 1$.

\section{Conclusions}
\label{sec:conc}
In this paper, we investigated the CF and PBA problem for imperfect NOMA in DL-HetNets. A practical SIC receiver is considered with real-life constraints and limitations including power disparity, sensitivity, and residual interference due to the error propagation, which is shown to have a significant impact on NOMA performance. After putting the cell users into a multi-partite graph, a sequential weighted bi-partite matching approach is developed for a fat and yet high-performance CF. Thereafter, we proposed a hierarchical distributed PBA scheme where the slave, secondary master, and primary master problems are responsible for UEs' power allocation, clusters' power allocation, and clusters' bandwidth allocation, respectively. Closed-form optimal power allocations and Lagrange multipliers of imperfect NOMA scheme are derived for slave problems, which are then used by the secondary master problem to update power allowances of each cluster. Numerical results show that the proposed distributed PBA scheme performs very close to the centralized benchmark. 
\appendices
\begin{figure*}
\begin{align}
\nonumber \omega_{1}=&\frac{\varpi}{q_{2}q_{3}q_{4}}
-\frac{(q_{2}-1)\rho_{2}}{q_{2}}
-\frac{(q_{3}-1)\rho_{3}}{q_{2}q_{3}} -\frac{(q_{4}-1)\rho_{4}}{q_{2}q_{3}q_{4}} 
-\frac{\epsilon_{2}\varpi(q_{2}-1)}{q_{2}q_{3}q_{4}}  -\frac{\epsilon_{2}\varpi(q_{2}-1)(q_{4}-1)}{q_{2}q_{3}q_{4}}\\
\nonumber &-\frac{\epsilon_{2}\varpi(q_{2}-1)(q_{3}-1)(q_{4}-1)}{q_{2}q_{3}q_{4}}
-\frac{\epsilon_{2}(q_{2}-1)(q_{4}-1)\rho_{4}}{q_{2}q_{3}q_{4}} -\frac{\epsilon_{2}(q_{2}-1)(q_{3}-1)\rho_{3}}{q_{2}q_{3}q_{4}}
-\frac{\epsilon_{2}(q_{2}-1)(q_{3}-1)(q_{4}-1)\rho_{3}}{q_{2}q_{3}q_{4}}\\
\label{eq:app1}
 &-\frac{\epsilon_{3}(q_{3}-1)(q_{4}-1)\rho_{4}}{q_{2}q_{3}q_{4}}-\frac{\epsilon_{3}\varpi(q_{3}-1)(q_{4}-1)}{q_{2}q_{3}q_{4}}-\frac{\epsilon_{2}\epsilon_{3}(q_{2}-1)(q_{3}-1)(q_{4}-1)\rho_{4}}{q_{2}q_{3}q_{4}} - \frac{\epsilon_{2}\epsilon_{3}\varpi(q_{2}-1)(q_{3}-1)(q_{4}-1)}{q_{2}q_{3}q_{4}}\\
\nonumber \omega_{2}=&\frac{\varpi(q_{2}-1)}{q_{2}q_{3}q_{4}}
+\frac{(q_{2}-1)\rho_{2}}{q_{2}}
-\frac{(q_{2}-1)(q_{3}-1)\rho_{3}}{q_{2}q_{3}}  -\frac{(q_{2}-1)\rho_{4}}{q_{2}q_{3}q_{4}}
+\frac{\epsilon_{2}\varpi(q_{2}-1)}{q_{2}q_{3}q_{4}}+\frac{\epsilon_{2}\varpi(q_{2}-1)(q_{4}-1)}{q_{2}q_{3}q_{4}}\\
\nonumber &
+\frac{\epsilon_{2}\varpi(q_{2}-1)(q_{3}-1)(q_{4}-1)}{q_{2}q_{3}q_{4}}
+\frac{\epsilon_{2}(q_{2}-1)(q_{4}-1)\rho_{4}}{q_{2}q_{3}q_{4}}
+\frac{\epsilon_{2}(q_{2}-1)(q_{3}-1)\rho_{3}}{q_{2}q_{3}q_{4}}
+\frac{\epsilon_{2}(q_{2}-1)(q_{3}-1)(q_{4}-1)\rho_{3}}{q_{2}q_{3}q_{4}}\\
\nonumber &
-\frac{\epsilon_{3}(q_{2}-1)(q_{3}-1)(q_{4}-1)\rho_{4}}{q_{2}q_{3}q_{4}}-\frac{\epsilon_{3}\varpi(q_{2}-1)(q_{3}-1)(q_{4}-1)}{q_{2}q_{3}q_{4}}+\frac{\epsilon_{2}\epsilon_{3}(q_{2}-1)(q_{3}-1)(q_{4}-1)\rho_{4}}{q_{2}q_{3}q_{4}}\\& +\frac{\epsilon_{2}\epsilon_{3}\varpi(q_{2}-1)(q_{3}-1)(q_{4}-1)}{q_{2}q_{3}q_{4}}\\
 \nonumber \omega_{3}=&\frac{\varpi(q_{3}-1)}{q_{2}q_{4}}
+\frac{(q_{3}-1)\rho_{3}}{q_{3}}
-\frac{(q_{3}-1)(q_{4}-1)\rho_{4}}{q_{3}q_{4}}+\frac{\epsilon_{3}\varpi(q_{3}-1)(q_{4}-1)}{q_{3}q_{4}}
+\frac{\epsilon_{3}(q_{3}-1)(q_{4}-1)\rho_{4}}{q_{3}q_{4}}\\
\label{eq:app2}\omega_{4}=&\frac{\varpi(q_{4}-1)}{q_{4}}
+\frac{(q_{4}-1)\rho_{4} }{q_{4} }
\end{align}
\begin{align}
\label{eq:kappa}
\nonumber \kappa_{i-1}^\star-\kappa_{i}^\star&=\frac{ \rho_{i}-\rho_{i-1}+\epsilon_{i}\sum_{k=i+1}^{K}\omega_k^\star -\epsilon_{i-1}\sum_{k=i}^{K}  \omega_k^\star}{ \left(\sum_{j=1}^{i-1}{\omega_j^\star} +\epsilon_{i-1} \sum_{k=i}^{K}  \omega_k^\star+ \rho_{i-1} \right) \left(\sum_{j=1}^{i-1}{\omega_j^\star} +\epsilon_{i} \sum_{k=i+1}^{K}  \omega_k^\star+ \rho_{i}^\star \right)}\\
& +\frac{ \epsilon_{i-1}\omega_{i-1}^\star}{ \left({\sum_{j=1}^{i-2}\omega_j^\star} +\epsilon_{i-1} \sum_{k=i}^{K}  \omega_k^\star+ \rho_{i-1} \right) \left(\sum_{j=1}^{i-1}{\omega_j^\star} +\epsilon_{i-1} \sum_{k=i}^{K}  \omega_k^\star+ \rho_{i-1} \right)}
\end{align}
\hrule
\end{figure*}

\section{Proofs for Lemma \ref{lem:cfp} and Lemma \ref{lem:opt_lamda}}
\label{proof:KKT}
Without loss of generality, we eliminate the cell/cluster indices and focus on the verification of KKT conditions for a 4-UE cluster at two active set examples: $\Xi_1=\{\lambda,\mu_2, \mu_3, \mu_4\}$ and $\Xi_2=\{\lambda,\varphi_3, \varphi_3,\varphi_4\}$. In the case of $\Xi_1$, optimal power levels are obtained from $\{ \mathrm{S}_1,\mathrm{S}_2^2,\mathrm{S}_2^3,\mathrm{S}_2^4 \}$ while $\{ \mathrm{S}_2^1,\mathrm{S}_3^2,\mathrm{S}_3^3,\mathrm{S}_3^4 \}$ stay as set of conditions to be satisfied, i.e., $\{\lambda, \mu_2, \mu_3, \mu_4>0\}$ and $\{\mu_1=\varphi_2=\varphi_3=\varphi_4=0\}$. Therefore, solving equations $\{ \mathrm{S}_1,\mathrm{S}_2^2,\mathrm{S}_2^4,\mathrm{S}_2^4 \}$, optimal power levels can be obtained as in (\ref{eq:app1})-(\ref{eq:app2}).

Optimal power allocations for all combinations can be derived by following the same steps for different cluster sizes, which finally yields the generalized formulas given in Lemma \ref{lem:cfp}. Notice that Lemma \ref{lem:cfp} ensures the primal KKT conditions and reduces to perfect NOMA case in \cite{Ali2016} if all $\epsilon$ terms are set to zero. For verification of dual conditions and proving Lemma \ref{lem:opt_lamda}, let us  rewrite (\ref{eq:Lag_derv_omg2}) as follows
\begin{align}
 \label{eq:Lag_derv_omg11}
  \kappa_{1}^{\star} &= \lambda+  \mu_1 + \sum_{j=2}^{K} \varphi_j  +\sum_{k =2}^{K}\mu_k \epsilon_k \left( q_k -1\right) \\
 \nonumber & \kappa_i^\star = \lambda  -\varphi_i + \sum_{j=i+1}^{K} \varphi_j  -  \mu_i  + \sum_{j =1}^{i-1} \mu_j \epsilon_j \left( q_j -1\right)  \\
 \label{eq:Lag_derv_omg22}
  &+ \sum_{k =i+1}^{K}  \mu_k \left( q_k -1\right) , i \geq 2
\end{align}
where $\kappa_1^\star$ and $\kappa_i^\star$ are constants and can be obtained by substituting optimal power levels given by Lemma \ref{lem:cfp} into first two and three term of (\ref{eq:Lag_derv_omg2}), respectively. Nothing that  $\mu_{c,r}^{2,\star}=0$ when $\varphi_{c,r}^{2,\star}$ and vice versa, $\mu_{c,r}^{2,\star}$ and  $\varphi_{c,r}^{2,\star}$ is obtained from subtracting (\ref{eq:Lag_derv_omg22}) from  (\ref{eq:Lag_derv_omg11}). Accordingly, the proof of Lemma \ref{lem:opt_lamda} follows from applying this subtraction recursively and employing previously found optimal values at each step. To verify KKT conditions are satisfied, we need all Lagrange multipliers to be non-negative, this is satisfied by ensuring $\kappa_i^\star-\kappa_{i+1}^\star$ non-negative. After some manipulations $\kappa_i^\star-\kappa_{i+1}^\star$ can be written as in \eqref{eq:kappa} which is indeed non-negative thanks to the descending ordering of the $\omega_i$ and $\rho_i$ values and non-negative terms, i.e., $\omega_i^\star \geq 0$, $\rho_i \geq 0$,  and $q_i\geq 0$.
\section{ Convexity Analysis of Primary Master Problem}
\label{proof:conv}
Without loss of generality, let us omit the iteration indices and consider single cell. While total bandwidth allocation constraint of the primary master problem is affine, QoS constraints are concave as the second derivative of  the achievable rate, $\frac{\partial^2}{\partial (\theta_c^r)^2}\theta_c^r \log_2\left( 1+\gamma_{c,r}^i(\theta_c^r)\right) \leq 0$, is always non-positive due to non-negative optimal powers and parameters. However, the objective function is the summation of individual rates which are coupled by cluster bandwidth, which is not always negative and thus not convex. Therefore, employing the utility of the previous iteration relaxes this highly non-convex objective into a linear function of bandwidths.     
\bibliographystyle{IEEEtran}

\bibliography{arxiv}

\vspace*{-3\baselineskip}
\begin{IEEEbiography}[{\includegraphics[width=1.1in,height=1.25in]{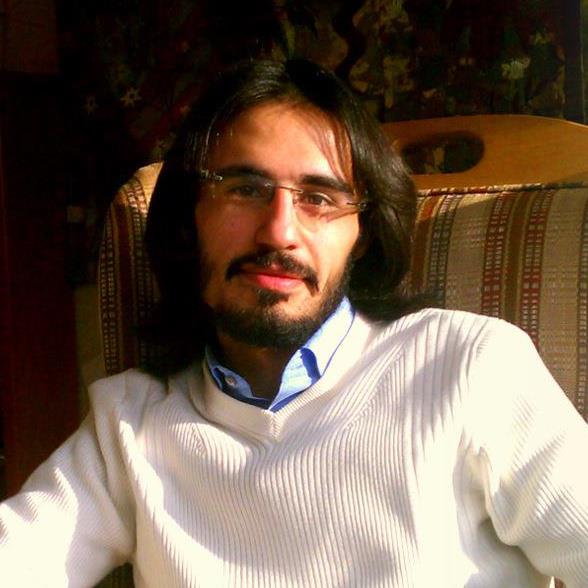}}]{Abdulkadir Celik}
(S'14-M'16) received the B.S. degree in electrical-electronics engineering from Selcuk University, Konya, Turkey in 2009, the M.S. degree in electrical engineering in 2013, the M.S. degree in computer engineering in 2015, and the Ph.D. degree in co-majors of electrical engineering and computer engineering in 2016, all from Iowa State University, Ames, IA, USA. He is currently a postdoctoral research fellow at Communication Theory Laboratory of King Abdullah University of Science and Technology (KAUST). His current research interests include but not limited to 5G and beyond, wireless data centers, UAV assisted cellular and IoT networks, and underwater optical wireless communications, networking, and localization.
\end{IEEEbiography}
\newpage 
\begin{IEEEbiography}[{\includegraphics[width=1.1in,height=1.25in]{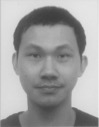}}]{Ming-Cheng Tsai}(S'18) was born in Fujian, China. He received his B.E. degree in electrical engineering from the National Taipei University of Technology, Taipei, Taiwan, in 2015. From 2015 to 2018, he was a student of Communication Engineering at the National Tsinghua University. He is currently pursuing his M.S./Ph.D. degree in Electrical Engineering at King Abdullah University of Science and Technology. His research interests include device to device communication, digital communication, and error correcting codes. 
\end{IEEEbiography}
\vspace*{-3\baselineskip}
\begin{IEEEbiography}[{\includegraphics[width=1.1in,height=1.25in]{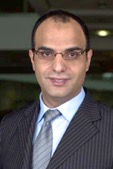}}]{Redha M. Radaydeh}(S'05-M'07-SM'13) was born in Irbid, Jordan, on November 12, 1978. He received the B.S. and M.S. degrees from Jordan University of Science and Technology (JUST), Irbid, in 2001 and 2003, respectively, and the Ph.D. degree from University of Mississippi, Oxford, MS, USA, in 2006, all in electrical engineering. He worked at JUST, King Abdullah University of Science and Technology (KAUST), Texas A\&M University at Qatar (TAMUQ), and Alfaisal University as an associate professor of electrical engineering. He also worked as a remote research scientist with KAUST and was a visiting researcher with Texas A\& M University (TAMU), College Station, TX. Currently, he is a faculty member with the electrical engineering program at Texas A\&M University-Commerce (TAMUC), Commerce, TX. His research interests include broad topics on wireless communications, and design and performance analysis of wireless networks.
\end{IEEEbiography}
\vspace*{-3\baselineskip}
\begin{IEEEbiography}[{\includegraphics[width=1.1in,height=1.25in]{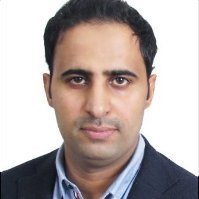}}]{Fawaz S. Al-Qahtani}
(M'10) received the B.Sc. in electrical engineering from King Fahad University of Petroleum and Minerals (KFUPM), Saudi Arabia in 2000 and M.Sc. in Digital Communication Systems from Monash University, Melbourne, Australia in 2005, and Ph.D. degree in Electrical and Computer Engineering, from RMIT University, Australia in December, 2009. He is currently working as IP commercialization manager for ICT portfolio at Research, Development, and Innovation (RDI) in Qatar Foundation, Doha, Qatar. From May 2010 to August 2017, he was research scientist with Texas A\&M University at Qatar, Education City, Doha, Qatar. His research has been sponsored by Qatar National Research Fund (QNRF). He was awarded of JSERP and NPRP projects. His current research interests include channel modeling, applied signal processing, MIMO communication systems, cooperative communications, cognitive radio systems, free space optical, physical layer security, and device-to-device communication. He is the author or co-author of 100 technical papers published in scientific journals and presented at international conferences.
\end{IEEEbiography}
\vspace*{-3\baselineskip}
\begin{IEEEbiography}[{\includegraphics[width=1.1in,height=1.25in]{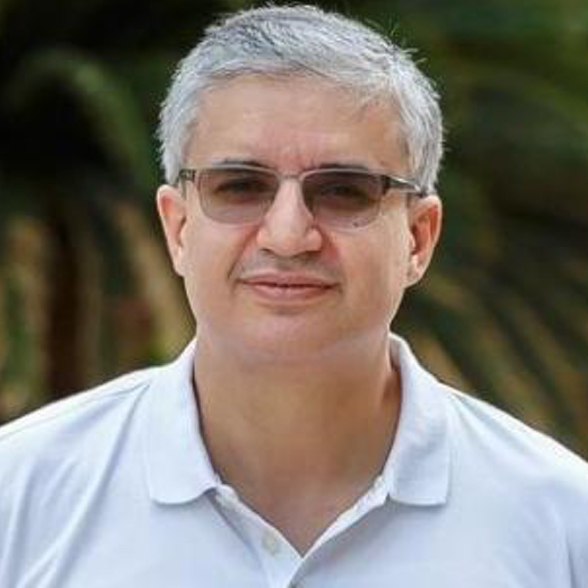}}]{Mohamed-Slim Alouini}
(S'94-M'98-SM'03-F'09) was born in Tunis, Tunisia. He received the Ph.D. degree in Electrical Engineering from the California Institute of Technology (Caltech), Pasadena, CA, USA, in 1998. He served as a faculty member in the University of Minnesota, Minneapolis, MN, USA, then in the Texas A\&M University at Qatar, Education City, Doha, Qatar before joining King Abdullah University of Science and Technology (KAUST), Thuwal, Makkah Province, Saudi Arabia as a Professor of Electrical Engineering in 2009. His current research interests include the modeling, design, and performance analysis of wireless communication systems.
\end{IEEEbiography}

\end{document}